\documentclass[aps,12pt,twoside,notitlepage,superscriptaddress]{revtex4-1}

\usepackage{graphicx,epic,eepic,epsfig,amsmath,amsfonts,amsthm,latexsym,amssymb,color,enumerate}
\allowdisplaybreaks[1]
\usepackage[colorlinks,
            linkcolor=blue,
            anchorcolor=blue,
            citecolor=red,
            ]{hyperref}

\newtheorem{theorem}{Theorem}

\newtheorem{corollary}[theorem]{Corollary}

\newtheorem{lemma}[theorem]{Lemma}

\newtheorem{proposition}[theorem]{Proposition}

\theoremstyle{definition}
\newtheorem{definition}[theorem]{Definition}
\theoremstyle{remark}
\newtheorem{remark}[theorem]{Remark}

\newenvironment{Proof}[1][Proof]{\noindent\textit{Proof.} }{\hfill\qed}
\newenvironment{Proofof}[1][Proof]{\noindent\textit{Proof of~#1.} }{\hfill\qed}

\newcommand{\nc}{\newcommand}
\nc{\rnc}{\renewcommand}
\nc{\beq}{\begin{equation}}
\nc{\eeq}{\end{equation}}
\nc{\bsp}{\begin{split}}
\nc{\esp}{\end{split}}
\nc{\beqa}{\begin{eqnarray}}
\nc{\eeqa}{\end{eqnarray}}
\nc{\lbar}[1]{\overline{#1}}
\nc{\ket}[1]{|#1\rangle}
\nc{\bra}[1]{\langle#1|}
\nc{\braket}[2]{\langle #1 | #2 \rangle}
\nc{\ketbra}[2]{|#1\rangle\!\langle#2|}
\nc{\proj}[1]{| #1\rangle\!\langle #1 |}
\nc{\avg}[1]{\langle#1\rangle}
\nc{\Rank}{\operatorname{Rank}}
\nc{\smfrac}[2]{\mbox{$\frac{#1}{#2}$}}
\nc{\tr}{\operatorname{Tr}}
\nc{\ox}{\otimes}
\nc{\dg}{\dagger}
\nc{\dn}{\downarrow}
\nc{\supp}{{\operatorname{supp}}}
\nc{\qsupp}{{\operatorname{qsupp}}}
\nc{\var}{\operatorname{var}}
\nc{\rar}{\rightarrow}
\nc{\lrar}{\longrightarrow}
\nc{\polylog}{\operatorname{polylog}}
\nc{\1}{{\openone}}
\nc{\id}{{\operatorname{id}}}
\nc{\<}{\langle}
\nc{\Hom}[2]{\mbox{Hom}(\CC^{#1},\CC^{#2})}
\nc{\rU}{\mbox{U}}
\nc{\mc}{\mathcal}

\begin{document}

\title{Reliability Function of Quantum Information Decoupling \\ via the Sandwiched R\'enyi Divergence}

\author{Ke Li}
  \email{carl.ke.lee@gmail.com}
  \affiliation{Institute for Advanced Study in Mathematics, Harbin Institute of Technology, Harbin
                 150001, China}
\author{Yongsheng Yao}
  \email{yongsh.yao@gmail.com}
  \affiliation{Institute for Advanced Study in Mathematics, Harbin Institute of Technology, Harbin
                 150001, China}	
  \affiliation{School of Mathematics, Harbin Institute of Technology, Harbin 150001, China}
	

\begin{abstract}
Quantum information decoupling is a fundamental quantum information processing
task, which also serves as a crucial tool in a diversity of topics in quantum
physics. In this paper, we characterize the reliability function of catalytic
quantum information decoupling, that is, the best exponential rate under which
perfect decoupling is asymptotically approached. We have obtained the exact
formula when the decoupling cost is below a critical value. In the situation
of high cost, we provide meaningful upper and lower bounds. This result is then
applied to quantum state merging, exploiting its inherent connection to decoupling.
In addition, as technical tools, we derive the exact exponents for the smoothing
of the conditional min-entropy and max-information, and we prove a novel bound
for the convex-split lemma.

Our results are given in terms of the sandwiched R\'enyi divergence, providing
it with a new type of operational meaning in characterizing how fast the
performance of quantum information tasks approaches the perfect.
\end{abstract}

\maketitle

\section{Introduction}
\label{sec:introduction}
\emph{Quantum information decoupling}~\cite{HOW2005partial,HOW2007quantum,
ADHW2009mother} is the procedure of removing the information of a reference
system from the system under control, via physically permitted operations.
It is a fundamental quantum information processing task, which has found
broad applications, ranging from quantum Shannon theory~\cite{HOW2005partial,
HOW2007quantum,ADHW2009mother,DevetakYard2008exact,YardDevetak2009optimal,
BDHSW2014quantum,BCR2011the,BBMW2018conditional,BFW2013quantum} to quantum thermodynamics~\cite{DARDV2011thermodynamic,BrandaoHorodecki2013area,
BrandaoHorodecki2015exponential,Aberg2013truly,DHRW2016relative} to black-hole
physics~\cite{HP2007black,BraunsteinPati2007quantum,BSZ2013better,
KamilChristoph2016one}. Since being introduced in~\cite{HOW2005partial,
HOW2007quantum,ADHW2009mother}, the problem of quantum information decoupling
has attracted continued interest of study from the community. This includes
the study of decoupling in the one-shot setting~\cite{BCR2011the,DBWR2014one,
ADJ2017quantum,MBDRC2017catalytic,MBGYA2019a,WakakuwaNakata2021one,Dupuis2021privacy},
the search for more specific and more efficient decoupling operations~\cite{SDTR2013decoupling,BrownFawzi2015decoupling,NHMW2017decoupling},
and the investigation of the speed of asymptotic convergence of the decoupling performance~\cite{Sharma2015random,MBGYA2019a,ABJT2020partially,Dupuis2021privacy}.
In particular, by introducing an independent system as catalyst, tight one-shot
characterization has been derived in~\cite{MBDRC2017catalytic,ADJ2017quantum},
which is able to provide the exact second-order asymptotics.

The \emph{reliability function} was introduced by Shannon in information theory~\cite{Shannon1959probability}. Defined as the rate of exponential decay of
the error with the increasing of blocklength, the reliability function provides
the desired precise characterization of how rapidly an information processing task
approaches the perfect in the asymptotic setting~\cite{Gallager1968information}.
Study of reliability functions in quantum information dates back to the work of
Burnashev and Holevo~\cite{BurnashevHolevo1998on,Holevo2000reliability}, and
Winter~\cite{Winter1999coding}, more than two decades ago. In recent years,
there has been a growing body of research in this topic from the quantum community~\cite{Dalai2013lower,Hayashi2015precise,
DalaiWinter2017constant,CHT2019quantum,CHDH2020non,Dupuis2021privacy}. However,
complete characterization of the reliability functions in the quantum regime is
not known, even for classical-quantum channels. Nevertheless, see References~\cite{KoenigWehner2009strong,SharmaWarsi2013fundamental,WWY2014strong,
GuptaWilde2015multiplicativity,CMW2016strong,MosonyiOgawa2017strong,CHDH2020non}
for a partial list of the fruitful results on the strong converse exponent in
the quantum setting, which characterizes how fast a quantum information task
becomes the useless.

In this paper, we investigate the reliability function for the task of
quantum information decoupling in the catalytic setting. We have obtained
the exact formula when the decoupling cost is below a critical value.
Specifically, for a bipartite quantum state $\rho_{RA}$, we consider three
different types of decoupling operations on the $A$ system:
(a) decoupling via removing a subsystem,
(b) decoupling via projective measurement,
(c) decoupling via random unitary operation.
We show that under any of these three types of decoupling operations,
the reliability function is given by the Legendre transformation of the
\emph{sandwiched R\'enyi mutual information}
\[
  I_{\alpha}(R:A)_\rho :=
  \min_{\sigma_A\in\mc{S}(A)} D_{\alpha}(\rho_{RA} \| \rho_R \ox \sigma_A)
\]
of order $\alpha\in (1,2]$. Here $\mc{S}(A)$ is the set of all quantum
states on system $A$, and
\[D_{\alpha}(M\|N) :=\frac{1}{\alpha-1} \log \tr\big(N^{\frac{1-\alpha}{2\alpha}}MN^{\frac{1-\alpha}{2\alpha}}\big)^\alpha
\]
is the \emph{sandwiched R\'enyi divergence}~\cite{MDSFT2013on,WWY2014strong}.
This result is obtained by deriving respective upper and lower bounds, and
we show that the two bounds coincide when the rate of the decoupling cost is
below the critical value. On the one hand, we analyse the convex-split
lemma of~\cite{ADJ2017quantum} and derive for it a novel bound in terms of the
sandwiched R\'enyi divergence of order $\alpha\in(1,2]$, and this constitutes
the main technical tool for proving the lower bound. On the other hand, the
upper bound is obtained based on an asymptotic analysis of the smoothing
quantity of the max-information, for which we show that the exact exponent
is given by a formula in terms of the sandwiched R\'enyi mutual information
of order $\alpha\in(1,\infty)$. Furthermore, as application, we provide
similar characterization for the reliability function of quantum state
merging by exploiting the inherent connection between quantum state merging
and decoupling.

Our results, along with the concurrent work of~\cite{LYH2023tight} which
addresses different problems, have provided the sandwiched R\'enyi divergence~\cite{MDSFT2013on,WWY2014strong} with a new type of operational
interpretation by showing that it characterizes the exact exponents under which
certain quantum information tasks approach the perfect. This is in stark
contrast to what was previously known that the sandwiched R\'enyi divergence
characterizes the strong converse exponents---the optimal exponential rates
under which the underlying errors go to $1$~\cite{MosonyiOgawa2015quantum,
MosonyiOgawa2015two,CMW2016strong,HayashiTomamichel2016correlation,
MosonyiOgawa2017strong,CHDH2020non}. Therefore, we conclude that the meaning
of this fundamental entropic quantity can be more fruitful than what was
previously understood.

\emph{Relation to previous works.}
In References~\cite{Sharma2015random} and~\cite{Dupuis2021privacy}, exponential
achievability bounds for the decoupling error were given, which are in terms
of the sandwiched R\'enyi divergence of order $\alpha\in(1,2]$, too. However,
these bounds do not seem to be able to yield the optimal exponent in the
asymptotic setting, and thus do not provide much information on the reliability
function. In another remarkable work of Reference~\cite{MBGYA2019a}, the authors
derived a decoupling theorem based on the so called vector-valued $L_p$-norms,
which also leads to exponential achievability bounds in terms of the sandwiched
R\'enyi divergence of order $\alpha\in(1,2]$. As there is no discussion on the
converse part, we do not know how tight the bounds of~\cite{MBGYA2019a} are.
Another difference between the above-mentioned works and the present one is that,
in the works~\cite{Sharma2015random,MBGYA2019a,Dupuis2021privacy} the decoupling
error is measured using the trace distance or a new correlation measure based on
vector-valued $L_p$-norms, while in the present paper we employ the purified
distance, or equivalently, the fidelity function.

The remainder of this paper is organized as follows. In Section~\ref{sec:preliminaries}
we introduce the necessary notation, definitions and some basic properties. In
Section~\ref{sec:problem-results} we present the problem formulation, the main
results, and the application to quantum state merging. The proofs are given in
Section~\ref{sec:proof-main} and Section~\ref{sec:relations}, where in
Section~\ref{sec:proof-main} we prove the characterization of the reliability
functions, and in Section~\ref{sec:relations} we prove the relation between
different types of decoupling as well as the relation between decoupling and
quantum state merging. At last, in Section~\ref{sec:discussion} we conclude the
paper with some discussion and open questions.

\section{Preliminaries}
\label{sec:preliminaries}
\subsection{Notation and basic properties}
Let $\mc{H}$ be a Hilbert space, and $\mc{H}_A$ be the Hilbert space associated
with system $A$. $\mc{H}_{AB}$, denoting the Hilbert space of the composite
system $AB$, is the tensor product of $\mc{H}_A$ and $\mc{H}_B$. We restrict
ourselves to finite-dimensional Hilbert spaces throughout this paper. The
notation $|A|$ stands for the dimension of $\mathcal{H}_A$. We use $\1_A$
to denote the identity operator on $\mc{H}_A$. The notation $\supp(X)$ for
an operator $X$ is used for the support of $X$. The set of unitary operators
on $\mathcal{H}$ is denoted as $\mathcal{U}(\mathcal{H})$, and the set of
positive semidefinite operators on $\mc{H}$ is denoted as $\mathcal{P}(\mc{H})$.
The set of normalized quantum states and subnormalized quantum states on
$\mathcal{H}$ are denoted as $\mathcal{S}(\mathcal{H})$ and $\mathcal{S}_
{\leq}(\mathcal{H})$, respectively. That is,
\begin{align*}
\mc{S}(\mc{H})        &=\{\rho \in \mc{P}(\mc{H}) | \tr\rho=1 \},\\
\mc{S}_{\leq}(\mc{H}) &=\{\rho \in \mc{P}(\mc{H}) | \tr\rho \leq 1 \}.
\end{align*}
If the Hilbert space $\mc{H}$ is associated with system $A$, then the above
notations $\mc{U}(\mc{H})$, $\mc{P}(\mc{H})$, $\mc{S}(\mc{H})$ and
$\mc{S}_{\leq}(\mc{H})$ are also written as $\mc{U}(A)$, $\mc{P}(A)$,
$\mc{S}(A)$ and $\mc{S}_{\leq}(A)$, respectively. The discrete Weyl operators
on a $d$-dimensional Hilbert space $\mc{H}$ with an orthonormal basis
$\{\ket{a}\}_{a=0}^{d-1}$ are a collection of unitary operators
\[
W_{a,b}
=\sum_{c=0}^{d-1}\mathrm{e}^{\frac{2\pi\mathrm{i}bc}{d}}
 \ket{(a+c)\!\!\!\mod d}\bra{c},
\]
where $a,b\in \{0,1,\ldots,d-1\}$.

For $X, Y \in \mc{P}(\mc{H})$, we write $X \geq Y$ if $X-Y \in \mc{P}(\mc{H})$
and $X \leq Y$ if $Y-X \in \mc{P}(\mc{H})$. $\{X \geq Y\}$ is the spectral
projection of $X-Y$ corresponding to all non-negative eigenvalues. $\{X>Y\}$,
$\{X \leq Y\}$ and $\{X <Y \}$ are similarly defined.

We use the purified distance~\cite{GLN2005distance, TCR2009fully} to measure
the closeness of a pair of states $\rho, \sigma \in \mathcal{S}_{\leq}(\mathcal{H})$.
The purified distance is defined as $P(\rho,\sigma):=\sqrt{1-F^2(\rho,\sigma)}$,
where
\[
F(\rho,\sigma):=\tr\sqrt{\sqrt{\sigma}\rho\sqrt{\sigma}}
                              + \sqrt{(1-\tr\rho)(1-\tr \sigma)}
\]
is the fidelity. The Uhlmann's theorem~\cite{Uhlmann1976transition}, stated as
follows, will play a key role in later proofs. Let $\rho_{AB}\in \mc{S}_{\leq}
(\mc{H_{AB}})$ be a bipartite state, and let $\sigma_{A} \in \mc{S}_{\leq}
(\mc{H_{A}})$. Then there exists an extension $\sigma_{AB}$ of $\sigma_{A}$ such
that $P(\rho_{AB},\sigma_{AB}) = P(\rho_{A},\sigma_{A})$.

A quantum operation or quantum channel $\Phi$ is a linear, completely positive,
and trace-preserving~(CPTP) map acting on quantum states. We denote by $\Phi_
{A \rar B}$ a quantum operation from system $A$ to system $B$. The Stinespring
representation theorem~\cite{Stinespring1955positive} states that there is an
ancillary system $C$ in a pure state $\proj{0}_C$, a system $E$ and a unitary
$U_{AC \rar BE}$ such that $\Phi_{A \rar B}(\rho_A)=\tr_{E} \big(U(\rho_A \ox
\proj{0}_C)U^*\big)$. A quantum measurement is described by a set of positive
semidefinite operators $\{M_x\}_x$ such that $\sum_xM_x=\1$. It outputs $x$ with
probability $\tr(\rho M_x)$ when the underlying state is $\rho$. If a measurement
$\{Q_x\}_x$ is such that all the $Q_x$ are projections onto mutually orthogonal
subspaces, then it is called a projective measurement. We associate each
quantum measurement $\mc{M}=\{M_x\}_x$ with a measurement channel
$\Phi_\mc{M}:\rho\mapsto \sum_x(\tr\rho M_x)\proj{x}$, where
$\{\ket{x}\}$ is an orthonormal basis.

Let $\sigma$ be a self-adjoint operator on $\mc{H}$ with spectral projections
$Q_1, \ldots, Q_{v(\sigma)}$, where $v(\sigma)$ is the number of different
eigenvalues of $\sigma$. The associated pinching map $\mathcal{E}_\sigma$ is
defined as
\[
\mathcal{E}_\sigma : X \rar \sum_i Q_i X Q_i.
\]
The pinching inequality~\cite{Hayashi2002optimal} states that for any $X \in
\mc{P}(\mc{H})$,
\begin{equation}
\label{equ:pinchingineq}
X \leq v(\sigma)\mc{E}_\sigma(X).
\end{equation}

For $n \in \mathbb{N}$, let $S_n$ be the symmetric group of the permutations of
$n$ elements. The set of symmetric states and subnormalized symmetric states on
$\mc{H}_{A^n}$ are defined, respectively, as
\begin{align*}
  \mathcal{S}^{\rm{sym}}(A^n) & :=
  \big\{\rho_{A^n} | \rho_{A^n} \in \mathcal{S}(A^n),~W_{\pi} \rho_{A^n}
  W_{\pi}^*=\rho_{A^n},~\forall \pi \in S_n\big\}, \\
  \mathcal{S}^{\rm{sym}}_{\leq}(A^n) & :=
  \big\{\rho_{A^n} | \rho_{A^n} \in \mathcal{S}_{\leq}(A^n),~
  W_{\pi} \rho_{A^n}  W_{\pi}^*=\rho_{A^n},~\forall \pi \in S_n\big\},
\end{align*}
where $W_{\pi} : \ket{\psi_1} \ox \cdots \ox \ket{\psi_n} \mapsto
\ket{\psi_{\pi^{-1}(1)}} \ox \cdots \ox \ket{\psi_{\pi^{-1}(n)}}$ is the natural
representation of $\pi \in S_n$. The set of the symmetric states can be dominated
by a single symmetric state, in the sense of the following Lemma~\ref{lemma:unisym},
and two different constructions are given in~\cite{CKR2009postselection}
and~\cite{Hayashi2009universal}, respectively.
See~\cite[Appendix A]{MosonyiOgawa2017strong} for a detailed proof.
\begin{lemma}
\label{lemma:unisym}
For every Hilbert space $\mathcal{H}_A$ and $n\in\mathbb{N}$, there exists a
universal symmetric state $\omega^{(n)}_{A^n} \in \mc{S}^{\rm{sym}}(A^n)$, such
that for any $\rho_{A^n} \in \mc{S}^{\rm{sym}}(A^n)$ we have
\begin{align*}
\rho_{A^n} &\leq g_{n,|A|} \omega^{(n)}_{A^n}, \\
v(\omega^{(n)}_{A^n}) &\leq (n+1)^{|A|-1},
\end{align*}
where $g_{n,|A|} \leq (n+1)^{\frac{(|A|+2)(|A|-1)}{2}}$ and $v(\omega^{(n)}_{A^n})$
denotes the number of different eigenvalues of $\omega^{(n)}_{A^n}$.
\end{lemma}

Throughout this paper, the functions $\exp$ and $\log$ are with base $2$, and
$\ln$ is with base $e$.

\subsection{Quantum entropies and information divergences}
The sandwiched R\'enyi divergence has been introduced
in~\cite{MDSFT2013on,WWY2014strong} and is a quantum generalization of the
classical R\'enyi information divergence. For quantum states $\rho, \sigma
\in \mc{S}(\mc{H})$ and a parameter $\alpha\in (0,1)\cup (1,\infty)$, it is
defined as
\[
D_{\alpha}(\rho\|\sigma) :=\frac{1}{\alpha-1} \log \tr\big
(\sigma^{\frac{1-\alpha}{2\alpha}}\rho\sigma^{\frac{1-\alpha}{2\alpha}}\big)^\alpha
\]
if either $\alpha>1$ and $\supp(\rho)\subseteq\supp(\sigma)$ or $\alpha<1$ and $\supp(\rho)\not\perp\supp(\sigma)$, otherwise we set $D_{\alpha}(\rho\|\sigma)
=+\infty$.

For a bipartite quantum state $\rho_{AB} \in \mc{S}(AB)$ and $\alpha\in
(0,1)\cup (1,\infty)$, the sandwiched R\'enyi mutual information of order $\alpha$
is defined as~\cite{WWY2014strong, Beigi2013sandwiched}
\[
  I_{\alpha}(A:B)_\rho :=
  \min_{\sigma_B\in\mc{S}(B)} D_{\alpha}(\rho_{AB} \| \rho_A \ox \sigma_B),
\]
and we consider a version of the sandwiched R{\'e}nyi conditional
entropy~\cite{MDSFT2013on}
\[
  H_{\alpha}(A|B)_\rho:=
  -\min_{\sigma_B\in\mc{S}(B)}D_{\alpha}(\rho_{AB} \| \1_A \ox \sigma_B).
\]
when the system $B$ is of dimension $1$, we recover from the sandwiched
R{\'e}nyi conditional entropy the  R{\'e}nyi entropy $H_{\alpha}(A)_\rho
:=-D_{\alpha}(\rho_{A} \| \1_A )=\frac{1}{1-\alpha}\log\tr\rho_A^\alpha$.

The quantum relative entropy~\cite{Umegaki1954conditional}
\[
  D(\rho \|\sigma) :=
  \begin{cases}
  \tr(\rho(\log\rho-\log\sigma)) & \text{ if }\supp(\rho)\subseteq\supp(\sigma), \\
  +\infty                        & \text{ otherwise}
  \end{cases}
\]
of states $\rho$ and $\sigma$ is the limit of the sandwiched R\'enyi divergence
when $\alpha\rar 1$. In the case $\alpha\rar \infty$, we get the max-relative
entropy~\cite{Datta2009min}
\[
D_{\rm{max}}(\rho \| \sigma) := \inf\{ \lambda~|~\rho \leq 2^\lambda \sigma \}.
\]
The limits of $I_{\alpha}(A:B)_\rho$ and $H_{\alpha}(A|B)_\rho$ when
$\alpha\rar 1$ are the quantum mutual information and quantum conditional entropy,
respectively: for $\rho_{AB} \in \mc{S}(AB)$,
\begin{align*}
I(A:B)_\rho:&=\min_{\sigma_B\in\mc{S}(B)} D(\rho_{AB} \| \rho_A \ox \sigma_B)
=D(\rho_{AB} \| \rho_A \ox \rho_B) ,\\
H(A|B)_\rho:&=-\min_{\sigma_B\in\mc{S}(B)} D(\rho_{AB} \| \1_A \ox \sigma_B)
=-D(\rho_{AB} \| \1_A \ox \rho_B).
\end{align*}
The max-information~\cite{CBR2013smooth} for $\rho_{AB} \in \mc{S}(AB)$,
\[
I_{\rm{max}}(A:B)_\rho
:=\min_{\sigma_B\in\mc{S}(B)}D_{\rm{max}}(\rho_{AB} \| \rho_A \ox \sigma_B),
\]
emerges as the limit of $I_{\alpha}(A:B)_\rho$ when $\alpha\rar\infty$.

\section{Problem Statement and Main Results}
\label{sec:problem-results}
\subsection{Catalytic quantum information decoupling}
Let a bipartite quantum state $\rho_{RA}\in\mc{S}(RA)$ be given. Quantum
information decoupling is the procedure of removing the information about
system $R$ from the $A$ system, by performing a quantum operation on the
$A$ system. In catalytic quantum information decoupling, an auxiliary
system $A'$ in state that is independent of the state $\rho_{RA}$ can be
added as a catalyst during the decoupling operation. Readers are referred
to~\cite{MBDRC2017catalytic} for a detailed description of standard and
catalytic decoupling. We consider three different types of decoupling
operations: (a) decoupling via removing a subsystem~\cite{ADHW2009mother},
(b) decoupling via projective measurement~\cite{HOW2007quantum},
(c) decoupling via random unitary operation~\cite{GPW2005quantum}.

\emph{(a) decoupling via removing a subsystem.} For quantum state $\rho_{RA}$,
a catalytic decoupling scheme via removing a subsystem
consists of a catalytic system $A'$ in the state $\sigma_{A'}$ and an
isometry operation $U:\mc{H}_{AA'} \rar \mc{H}_{A_1A_2}$, where $A_2$ is
the system to be removed and $A_1$ is the remaining system. Without loss
of generality, we require that $|AA'|=|A_1A_2|$ and hence $U_{AA'\rar A_1A_2}$
becomes unitary. The cost of the decoupling is given by the number of qubits
that is removed, $\log |A_2|$. The performance is measured by the
purified distance between the remaining state $\tr_{A_2}U(\rho_{RA}\ox
\sigma_{A'})U^*$ and the nearest product state of the form $\rho_R \ox
\omega_{A_1}$. We are interested in the optimal performance when the cost,
namely, the number of removed qubits, is bounded.
\begin{definition}
  \label{def:perform-dec}
Let $\rho_{RA}\in \mc{S}(RA)$ be a bipartite quantum state. For a given
size of removed system $k\geq 0$ (in qubits), the optimal performance
of catalytic decoupling via removing a subsystem is given by
\begin{equation}
\label{eq:perform-dec}
P^{\rm dec}_{R:A}(\rho_{RA},k):=
\min P\big(\tr_{A_2}U(\rho_{RA}\ox \sigma_{A'})U^*,\rho_R \ox
\omega_{A_1}\big),
\end{equation}
where the minimization is over all system dimensions $|A'|$, $|A_1|$, $|A_2|$
such that $|AA'|=|A_1A_2|$ and $\log |A_2|\leq k$, all states $\sigma_{A'}\in
\mc{S}(A')$, $\omega_{A_1}\in\mc{S}(A_1)$, and all unitary operations
$U:\mc{H}_{AA'} \rar \mc{H}_{A_1A_2}$.
\end{definition}

\emph{(b) decoupling via projective measurement.} For quantum state $\rho_{RA}$,
a catalytic decoupling scheme via projective measurement consists of a catalytic
system $A'$ in the state $\sigma_{A'}$ and an projective measurement
$\mc{Q}=\{Q^x\}_{x=1}^m$ on the composite system $AA'$. The cost of the decoupling
is accounted by the number (in bits) of measurement outcomes, $\log m$. The
performance is  measured by the purified distance between the post-measurement
state $\sum_{x=1}^m Q^x_{AA'}(\rho_{RA}\ox\sigma_{A'})Q^x_{AA'}$ and the nearest
product state of the form $\rho_R \ox\omega_{AA'}$. We are interested in the
optimal performance when the cost, namely, the number of measurement outcomes,
is bounded.
\begin{definition}
  \label{def:perform-dec-m}
Let $\rho_{RA}\in \mc{S}(RA)$ be a bipartite quantum state. For a given
number of measurement outcomes $k\geq 0$ (in bits), the optimal performance
of catalytic decoupling via projective measurement is given by
\begin{equation}
\label{eq:perform-dec-m}
P^{\rm{dec\text{-}m}}_{R:A}(\rho_{RA},k):=
\min P\Big(\sum_{x=1}^m Q^x_{AA'}(\rho_{RA}\ox\sigma_{A'})Q^x_{AA'},
\rho_R \ox\omega_{AA'}\Big),
\end{equation}
where the minimization is over all system dimensions $|A'|$, all states
$\sigma_{A'}\in\mc{S}(A')$, $\omega_{AA'}\in\mc{S}(AA')$, and all projective
measurements $\{Q^x_{AA'}\}_{x=1}^m$ such that $\log m \leq k$.
\end{definition}

\emph{(c) decoupling via random unitary operation.} For quantum state $\rho_{RA}$,
a catalytic decoupling scheme via random unitary operation consists of a
catalytic system $A'$ in the state $\sigma_{A'}$ and a random unitary operation
$\Lambda_{AA'}: X\mapsto \frac{1}{m}\sum_{i=1}^{m} U_i X U_i^*$ acting on the
composite system $AA'$. The cost of the decoupling is accounted by the number
(in bits) of unitary operators in $\Lambda_{AA'}$, $\log m$. The performance is
measured by the purified distance between the resulting state $\Lambda_{AA'}
(\rho_{RA}\ox\sigma_{A'})$ and the nearest product state of the form $\rho_R
\ox\omega_{AA'}$. We are interested in the optimal performance when the cost,
namely, the number of unitary operators, is bounded.
\begin{definition}
  \label{def:perform-dec-u}
Let $\rho_{RA}\in \mc{S}(RA)$ be a bipartite quantum state. For a given
number of unitary operators $k\geq 0$ (in bits), the optimal performance
of catalytic decoupling via random unitary operation is given by
\begin{equation}
\label{eq:perform-dec-u}
P^{\rm{dec\text{-}u}}_{R:A}(\rho_{RA},k):=
\min P\Big(\Lambda_{AA'}(\rho_{RA} \ox \sigma_{A'}),\rho_R \ox\omega_{AA'}\Big),
\end{equation}
where the minimization is over all system dimensions $|A'|$, all states
$\sigma_{A'}\in \mc{S}(A')$, $\omega_{AA'} \in \mc{S}(AA')$ and all random
unitary operations $\Lambda_{AA'}(\cdot)=\frac{1}{m}\sum_{i=1}^{m} U_i (\cdot)
U_i^*$ with $U_i \in\mc{U}(AA')$ such that $\log m \leq k$.
\end{definition}

The reliability function of quantum information decoupling characterizes
the speed at which perfect decoupling can be approached in the asymptotic
setting, in which the underlying bipartite quantum state is in the form of
tensor product of $n$ identical copies. Specifically, it is the rate of
exponential decreasing of the optimal performance, as a function of the
cost.

\begin{definition}
  \label{def:reliability-dec}
Let $\rho_{RA}\in \mc{S}(RA)$ be a bipartite quantum state, and $r\geq 0$.
The reliability functions of catalytic quantum information
decoupling for the state $\rho_{RA}$, via the three different types of decoupling
operations described above, are defined respectively as
\begin{align}
E^{\rm dec}_{R:A}(\rho_{RA},r)         &:=\limsup_{n\rar\infty}\frac{-1}{n} \log
                      P^{\rm dec}_{R^n:A^n}\big(\rho_{RA}^{\ox n},nr\big),
                      \label{eq:reliability-dec}  \\
E^{\rm dec\text{-}m}_{R:A}(\rho_{RA},r)&:=\limsup_{n\rar\infty}\frac{-1}{n} \log
                      P^{\rm{dec\text{-}m}}_{R^n:A^n}\big(\rho_{RA}^{\ox n},nr\big),
                      \label{eq:reliability-dec-m} \\
E^{\rm dec\text{-}u}_{R:A}(\rho_{RA},r)&:=\limsup_{n\rar\infty}\frac{-1}{n} \log
                      P^{\rm{dec\text{-}u}}_{R^n:A^n}\big(\rho_{RA}^{\ox n},nr\big).
                      \label{eq:reliability-dec-u}
\end{align}
\end{definition}

\subsection{Main results}
At first, we show in the following Proposition~\ref{prop:dec-relation} equalities
that relate the optimal performances or the reliability functions for the three
different decoupling operations. With this, we are able to deal with them in a
unified way.
\begin{proposition}
\label{prop:dec-relation}
For $\rho_{RA} \in \mathcal{S}(RA)$ and $k,r \geq 0$, we have
\begin{align}
&P_{R:A}^{\emph{dec}}(\rho_{RA}, k)
=P_{R:A}^{\rm dec\text{-}m}(\rho_{RA}, 2k)
=P_{R:A}^{\rm dec\text{-}u}(\rho_{RA}, 2k), \label{eq:dec-relation-P} \\
&E_{R:A}^{\emph{dec}}(\rho_{RA}, r)
=E_{R:A}^{\rm dec\text{-}m}(\rho_{RA}, 2r)
=E_{R:A}^{\rm dec\text{-}u}(\rho_{RA}, 2r). \label{eq:dec-relation-E}
\end{align}
\end{proposition}

Then, we derive a one-shot achievability bound for the performance of
catalytic quantum information decoupling. The bound is given in terms
of the sandwiched R\'enyi information divergence. In doing so, we have
employed the convex-split lemma~\cite{ADJ2017quantum} as a key technical
tool. While we present the bound only for the case of decoupling via
removing a subsystem in the following Theorem~\ref{thm:dec-p-up},
similar results hold for the other two cases in light of
Proposition~\ref{prop:dec-relation}.
\begin{theorem}
  \label{thm:dec-p-up}
Let $\rho_{RA}\in \mc{S}(RA)$. For any $m\in\mathbb{N}$, $0<s\leq 1$ and
$\sigma_A\in\mc{S}(A)$, the optimal performance of decoupling $A$ from $R$
is bounded as
\[
P^{\rm dec}_{R:A}(\rho_{RA},\log m) \leq
\sqrt{\frac{v^s}{s}}\exp\Big\{-s\big(\log m-\frac{1}{2}D_{1+s}
     (\rho_{RA}\| \rho_R\ox\sigma_A)\big)\Big\},
\]
where $v$ is the number of distinct eigenvalues of $\rho_R\ox \sigma_A$.
\end{theorem}

Our main result is the characterization of reliability functions.
This is given in Theorem~\ref{thm:reliability-dec} for the case of decoupling
via removing a subsystem. Thanks to Proposition~\ref{prop:dec-relation},
similar results follow directly for the other two cases. We point out that
we have completely determined the reliability functions when the respective
cost is below a critical value.
\begin{theorem}
  \label{thm:reliability-dec}
Let $\rho_{RA}\in \mc{S}(RA)$ be a bipartite quantum state, and consider
the problem of decoupling quantum information in $A^n$ from the
reference system $R^n$ for the quantum state $\rho_{RA}^{\ox n}$. When
$r\leq R_{\rm critical}:=\frac{1}{2}\frac{\mathrm{d}}
{\mathrm{d}s} sI_{1+s}(R:A)_{\rho}\big|_{s=1}$, we have
\beq\label{eq:exp-dec}
E^{\rm dec}_{R:A}(\rho_{RA},r) = \max_{0\leq s\leq 1}
  \Big\{s\big(r-\frac{1}{2}I_{1+s}(R:A)_\rho\big)\Big\}.
\eeq
In general, we have
\begin{align}
E^{\rm dec}_{R:A}(\rho_{RA},r) &\geq \max_{0\leq s\leq 1}
  \Big\{s\big(r-\frac{1}{2}I_{1+s}(R:A)_\rho\big)\Big\}, \label{eq:exp-dec-low} \\
E^{\rm dec}_{R:A}(\rho_{RA},r) &\leq \ \sup_{s \geq 0}
  \ \, \Big\{s\big(r-\frac{1}{2}I_{1+s}(R:A)_\rho\big)\Big\}.\label{eq:exp-dec-up}
 \end{align}
\end{theorem}

The lower bound of Eq.~\eqref{eq:exp-dec-low} is a consequence of
Theorem~\ref{thm:dec-p-up}. For the upper bound of Eq.~\eqref{eq:exp-dec-up},
we first bound the optimal decoupling performance using the smoothing quantity
associated with the max-information, and then we derive the exact exponent in
smoothing the max-information. Eq.~\eqref{eq:exp-dec} is derived from the
combination of Eqs.~\eqref{eq:exp-dec-low} and \eqref{eq:exp-dec-up}.

The results presented in Theorem~\ref{thm:reliability-dec} are depicted
in Figure~\ref{fig:reliability}. Above the critical value, we are unable
to determine the formula for the reliability function. This is indeed a
common hard problem in the topic of reliability functions (see, e.g., References~\cite{Gallager1968information,Dalai2013lower,
HayashiTan2016equivocations,LYH2023tight}). More comments on
Theorem~\ref{thm:reliability-dec} can be found in Remark~\ref{rem:mainresult}.

\begin{figure}[ht]
  \includegraphics[width=9.5cm]{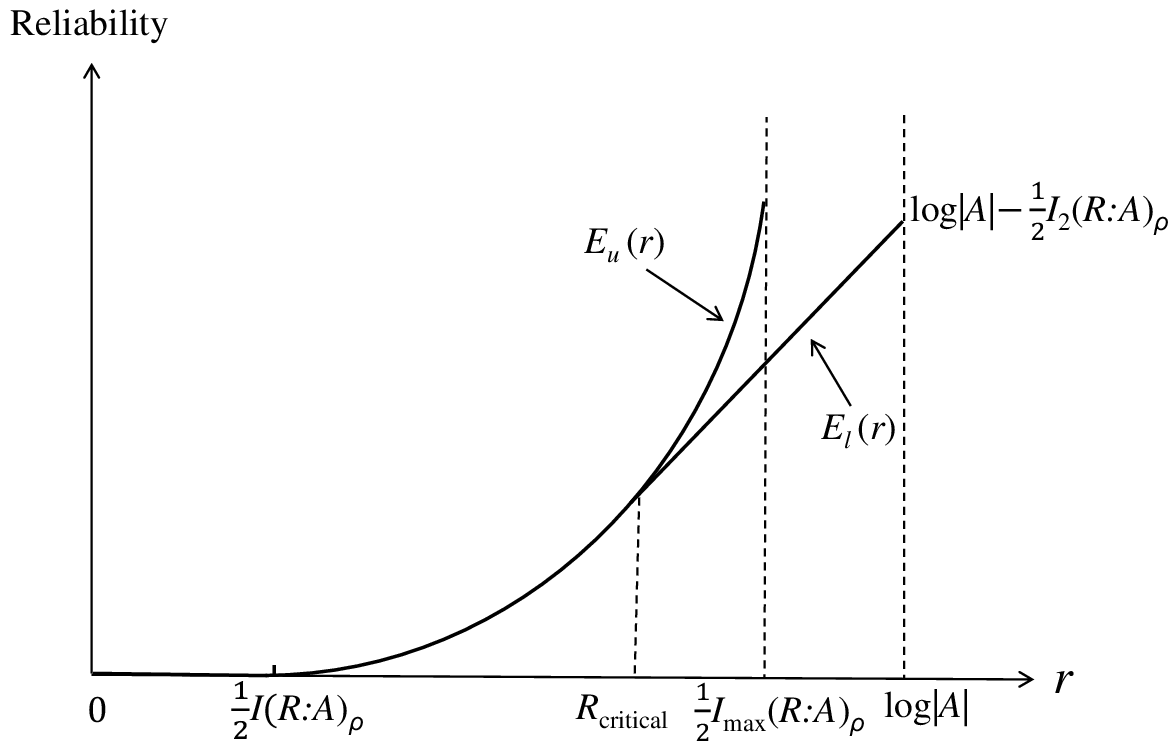}
  \caption{Reliability function of quantum information decoupling.
  $E_u(r):=\sup_{s \geq 0}
  \big\{s\big(r-\frac{1}{2}I_{1+s}(R:A)_\rho\big)\big\}$
  is the upper bound of Eq.~(\ref{eq:exp-dec-up}).
  $E_l(r):=\sup_{0\leq s\leq1}
  \big\{s\big(r-\frac{1}{2}I_{1+s}(R:A)_\rho\big)\big\}$
  is the lower bound of Eq.~(\ref{eq:exp-dec-low}). The
  two bounds are equal in the interval $[0,R_\text{critical}]$,
  giving the exact reliability function. The reliability function
  equals $0$ when $r<\frac{1}{2}I(R:A)_\rho$ and it is strictly
  positive when $r>\frac{1}{2}I(R:A)_\rho$. Above the critical value
  $R_\text{critical}$, the upper bound $E_u(r)$ becomes larger than
  the lower bound and it is $\infty$ when $r>\frac{1}{2}I_{\rm{max}}
  (R:A)_\rho$, while the lower bound $E_l(r)$ becomes linear and
  reaches $\log |A|-\frac{1}{2}I_2(R:A)_\rho$ at $r=\log |A|$.}
  \label{fig:reliability}
\end{figure}

The proof of Proposition~\ref{prop:dec-relation} can be found in
Section~\ref{sec:relations}, and those of Theorem~\ref{thm:dec-p-up}
and Theorem~\ref{thm:reliability-dec} are given in Section~\ref{sec:proof-main}.

\subsection{Applications to quantum state merging}
\label{subsec:merging}
The inherent connection between quantum information decoupling and quantum
state merging has been established since the invention of these two tasks~\cite{HOW2005partial,HOW2007quantum,ADHW2009mother}. It was further
explored later on, e.g., in~\cite{DBWR2014one,ADJ2017quantum,MBDRC2017catalytic,
ABJT2020partially}. Exploiting this connection, we are able to extend the
results on the reliability function of quantum information decoupling to
the scenario of quantum state merging.

Let $\rho_{RAB}$ be a tripartite pure state. Alice, Bob and a referee hold
system $A$, $B$ and $R$ respectively. Quantum state merging is the task of
transmitting the quantum information stored in the $A$ system from Alice to
Bob. There are two different ways to achieve this. One is by classical
communication and is introduced in~\cite{HOW2005partial}. The other one is
by quantum communication, firstly considered in~\cite{ADHW2009mother}.
In both cases, free pre-shared entanglement between Alice and Bob is allowed.
A formal description is as follows.

A quantum state merging protocol via quantum communication is represented by a
quantum operation $\mathcal{M}_1$, which
consists of using a shared bipartite entangled pure state $\phi_{A'B'}$, Alice
applying local unitary $U_{AA' \rightarrow A_1A_2}$ and sending the system $A_2$
to Bob, Bob applying local unitary $V_{A_2BB' \rightarrow ABB_1}$ and they
discarding the systems $A_1$ and $B_1$. A CPTP map $\mathcal{M}_2$ is a quantum
state merging protocol via classical communication if it consists of using a
shared bipartite entangled pure state $\psi_{A'B'}$, applying local operation at
Alice's side, sending $k$ classical bits from Alice to Bob and applying local
operation to reproduce systems $A$ and $B$ at Bob's side.
The performances of both protocols are given by the purified distance between
$\rho_{RAB}$ and the final state on the referee's system $R$ and Bob's system
$A$ and $B$. The cost of state merging that we are concerned with, is the number
of qubits ($\log |A_2|$ in $\mathcal{M}_1$) or classical bits ($k$ in $\mc{M}_2$)
that Alice sends to Bob.

\begin{definition}
Let $\rho_{RAB} \in \mathcal{S}({RAB})$ be a tripartite pure state and
$r \geq 0$. Let $P_{A \Rightarrow B}^{\rm{merg}}(\rho_{RAB}, k)$ denote
the optimal performance of quantum state merging via quantum communication of at
most $k$ qubits, and let $P_{A \rightarrow B}^{\rm{merg}}(\rho_{RAB}, k)$ denote
the optimal performance of quantum state merging via classical communication of
at most $k$ bits. They are defined respectively as
\begin{align}
\label{eq:p-merg-q}
&P_{A \Rightarrow B}^{\rm{merg}}(\rho_{RAB}, k)
:= \min_{\mathcal{M}_1} P(\mathcal{M}_1(\rho_{RAB}), \rho_{RAB} ),\\
\label{eq:p-merg-c}
&P_{A \rightarrow B}^{\rm{merg}}(\rho_{RAB}, k)
:= \min_{\mathcal{M}_2} P(\mathcal{M}_2(\rho_{RAB}), \rho_{RAB}),
\end{align}
where $\mathcal{M}_1$ is the protocol via quantum communication and the
minimization in Eq.~\eqref{eq:p-merg-q} is over all such protocols
whose communication cost is bounded by $k$ qubits, and $\mathcal{M}_2$
is the protocol via classical communication and the minimization in
Eq.~\eqref{eq:p-merg-c} is over all possible $\mathcal{M}_2$ whose
communication cost is bound by $k$ bits.
\end{definition}

The reliability function of quantum state merging characterizes the rate of
exponential decreasing of the optimal performance in the asymptotic limit.
\begin{definition}
Let $\rho_{RAB}\in\mathcal{S}({RAB})$ be a tripartite pure state and $r\geq 0$.
The reliability functions $E_{A \Rightarrow B}^{\rm{merg}}(\rho_{RAB}, r)$
and $E_{A \rightarrow B}^{\rm{merg}}(\rho_{RAB}, r)$ of quantum state merging
via quantum communication and classical communication respectively, are defined
as
\begin{align}
&E_{A \Rightarrow B}^{\rm{merg}}(\rho_{RAB}, r)
:=\limsup_{n\rightarrow\infty} \frac{-1}{n}
\log P_{A^n \Rightarrow B^n}^{\rm{merg}}(\rho_{RAB}^{\ox n}, nr),\\
&E_{A \rightarrow B}^{\rm{merg}}(\rho_{RAB}, r)
:=\limsup_{n\rightarrow\infty} \frac{-1}{n}
\log P_{A ^n\rightarrow B^n}^{\rm{merg}}(\rho_{RAB}^{\ox n}, nr).
\end{align}
\end{definition}

To apply our results on decoupling to the problem of quantum state merging,
we show exact equalities relating the optimal performances (or reliability
functions) of catalytic decoupling and quantum state merging.
Eq.~\eqref{perf-merg-dec} in Proposition~\ref{prop:merging-dec} is
essentially due to Uhlmann's theorem and is implicitly used in previous
works~\cite{HOW2005partial,HOW2007quantum,ADHW2009mother,DBWR2014one,
ADJ2017quantum,MBDRC2017catalytic,ABJT2020partially}. However, our
definition of the optimal decoupling performance
(cf. Definition~\ref{def:perform-dec}) is subtle to enable such an
equality relation. See Section~\ref{sec:relations} for the detailed proof.

\begin{proposition}
\label{prop:merging-dec}
For a tripartite pure state $\rho_{RAB} \in \mathcal{S}({RAB})$ and
$k\geq 0$, $r \geq 0$, we have
\begin{align}
&P^{\rm{merg}}_{A \Rightarrow B}(\rho_{RAB}, k)
=P^{\rm{merg}}_{A \rightarrow B}(\rho_{RAB}, 2k)
=P^{\rm{dec}}_{R:A}(\rho_{RA}, k), \label{perf-merg-dec} \\
&E^{\rm{merg}}_{A \Rightarrow B}(\rho_{RAB}, r)
=E^{\rm{merg}}_{A \rightarrow B}(\rho_{RAB}, 2r)
=E^{\rm{dec}}_{R:A}(\rho_{RA}, r). \label{reliab-merg-dec}
\end{align}
\end{proposition}

With Proposition~\ref{prop:merging-dec}, we immediately obtain results
in analogy to Theorem~\ref{thm:dec-p-up} and Theorem~\ref{thm:reliability-dec}
for quantum state merging. We do not lay them out entirely but only exhibit
the following.

\begin{corollary}
  \label{cor:reliability-merging}
Let $\rho_{RAB}\in\mathcal{S}({RAB})$ be a tripartite pure state. When the
rate of qubits transmission $r$ is such that $0\leq r \leq R_{\rm critical}
\equiv\frac{1}{2}\frac{\mathrm{d}}{\mathrm{d}s} sI_{1+s}(R:A)_{\rho}\big|_{s=1}$,
the reliability function of quantum state merging via quantum communication
is given by
\beq\label{eq:exp-merging-q}
E^{\rm{merg}}_{A \Rightarrow B}(\rho_{RAB}, r)
=\max_{0\leq s\leq 1}\Big\{s\big(r-\frac{1}{2}I_{1+s}(R:A)_\rho\big)\Big\};
\eeq
when the rate of classical bits transmission $r$ is such that $0\leq
\frac{r}{2}\leq R_{\rm critical}\equiv\frac{1}{2}\frac{\mathrm{d}}{\mathrm{d}s}
sI_{1+s}(R:A)_{\rho}\big|_{s=1}$, the reliability function of quantum state
merging via classical communication is given by
\beq\label{eq:exp-merging-c}
E^{\rm{merg}}_{A \rightarrow B}(\rho_{RAB}, r)
=\max_{0\leq s\leq 1}\Big\{\frac{s}{2}\big(r-I_{1+s}(R:A)_\rho\big)\Big\}.
\eeq
\end{corollary}

A special case where Bob holds no side information is worth looking at. It
can be understood as entanglement-assisted quantum source coding. In this case,
Alice and a referee share a pure state $\rho_{RA}$ and Alice wants to send
the information in the $A$ system to Bob with the assistance of unlimited
entanglement, using noiseless quantum or classical communication. It has
been shown in~\cite[Lemma 8]{WWW2019quantifying}
(see also~\cite[Proposition 4.86]{KhatriWilde2020principles}) that
\[
I_{1+s}(R:A)_\rho=\frac{2s+1}{s}\log\tr \rho_A^{\frac{1}{2s+1}}
                 =2H_{\frac{1}{2s+1}}(\rho_A),
\]
which can also be verified by employing the dual relation
of~\cite[Lemma 6]{HayashiTomamichel2016correlation}.
So the reliability function in this case has a simpler
formula involving only the R\'enyi entropy of one single system.

\section{Proof of the Characterization of Reliability Functions}
  \label{sec:proof-main}
In this section, we prove Theorem~\ref{thm:reliability-dec} on the
characterization of the reliability functions. As an intermediate step,
we also derive the one-shot achievability bound of Theorem~\ref{thm:dec-p-up}.
The proof is organized as follows. In Section~\ref{subsec:convex-split}, we
analyse the convex-split lemma~\cite{ADJ2017quantum}, obtaining a new bound
that employs the sandwiched R{\'e}nyi divergence. This will be a crucial tool
for proving Theorem~\ref{thm:reliability-dec}, regarding the achievability
bound of Eq.~\eqref{eq:exp-dec-low}. Then in Section~\ref{subsec:MI-CE}, we
derive the exact exponent for the asymptotic decreasing of the smoothing
quantity for the max-information. This will serve as another key technical
tool for proving Theorem~\ref{thm:reliability-dec}, for the converse bound of
Eq.~\eqref{eq:exp-dec-up}. At last, in Section~\ref{subsec:thm-proof}, we
accomplish the proof of Theorem~\ref{thm:reliability-dec} as well as
Theorem~\ref{thm:dec-p-up}, by employing the established tools mentioned above.

\subsection{A convex-split lemma}
  \label{subsec:convex-split}
The convex-split lemma was introduced in~\cite{ADJ2017quantum} and has broad
applications in topics such as one-shot quantum Shannon theory~\cite{ADJ2017quantum,
MBDRC2017catalytic,AJW2017generalized,AJW2018building,AJW2019convex}, entanglement
and general resource theories~\cite{AHJ2018quantifying,BertaMajenz2018disentanglement},
and quantum thermodynamics~\cite{FBB2021thermodynamic,LipkaSkrzypczyk2021all}.
Roughly speaking, it quantifies how well the information of a distinct object
located among many other identical ones, can be erased by randomly mixing all of
them. Originally in~\cite{ADJ2017quantum} and in all of the previous applications,
the effect of this erasure is bounded using the max-relative entropy, which, after
being smoothed, is sufficiently tight for those purposes. However, for our purpose
of deriving the reliability function, this bound does not work any more. Instead,
we prove a version of the convex-split lemma with a new bound, employing directly
the sandwiched R{\'e}nyi divergence.
\begin{lemma}
  \label{lem:convex-split}
Let $\rho_{RA}\in \mc{S}(RA)$ and $\sigma_A\in \mc{S}(\mc{H}_A)$
be quantum states such that $\supp (\rho_A)\subseteq\supp (\sigma_A)$.
Consider the following state
\[
\tau_{RA_1A_2\cdots A_m}:=\frac{1}{m}\sum_{i=1}^m
                 \rho_{RA_i}\ox \big[\sigma^{\ox (m-1)}\big]_{A^m/A_i},
\]
where $A^m/A_i$ denotes the composite system consisting of $A_1,A_2,\cdots,
A_{i-1},A_{i+1},\cdots,A_m$ and $\big[\sigma^{\ox (m-1)}\big]_{A^m/A_i}$
is the product state $\sigma^{\ox (m-1)}$ on these $m-1$ systems. Let
$v=v(\rho_R \ox \sigma_A)$ denote the number of distinct eigenvalues of
$\rho_R \ox \sigma_A$. Then for any $0<s\leq 1$,
\[
D\left(\tau_{RA_1A_2\cdots A_m}\big\| \rho_R \ox (\sigma^{\ox m})_{A^m}\right)
\leq \frac{v^s}{(\ln2)s}\exp\Big\{-s\big(\log m-D_{1+s}
     (\rho_{RA}\| \rho_R\ox\sigma_A)\big)\Big\}.
\]
\end{lemma}

\begin{Proof}
We use the shorthand $\xi_i \equiv \rho_{RA_i}\ox \big[\sigma^{\ox (m-1)}
\big]_{A^m/A_i}$ for simplicity and start with
\begin{align*}
     & D\Big(\tau_{RA_1A_2\cdots A_m}\big\| \rho_R \ox (\sigma^{\ox m})_{A^m}\Big) \\
   = &\tr\Big[\Big(\frac{1}{m}\sum_{i}\xi_i\Big) \Big(\log\big(\frac{1}{m}
      \sum_{i}\xi_i\big)-\log\big(\rho_R \ox (\sigma^{\ox m})_{A^m}\big)\Big)\Big] \\
   = &\tr\Big[\xi_1\Big(\log\big(\frac{1}{m}\sum_{i}\xi_i\big)-
      \log\big(\rho_R \ox (\sigma^{\ox m})_{A^m}\big)\Big)\Big]                \\
   = &\tr\Big[\xi_1\Big(\log\xi_1-\log\big(\rho_R\ox (\sigma^{\ox m})_{A^m}\big)\Big)\Big]
      -\tr\Big[\xi_1\Big(\log\xi_1-\log\big(\frac{1}{m}\sum_{i}\xi_i\big)\Big)\Big] \\
   = &D\Big(\rho_{RA_1}\ox (\sigma^{\ox (m-1)})_{A_2\cdots A_m}\big\|\rho_R\ox (\sigma^{\ox m})_{A^m}\Big)
     - D\Big(\rho_{RA_1}\ox (\sigma^{\ox (m-1)})_{A_2\cdots A_m}\big\| \frac{1}{m}\sum_{i}\xi_i\Big) \\
\leq & D\big(\rho_{RA_1} \| \rho_R \ox \sigma_{A_1}\big) - D\Big(\rho_{RA_1}\big\|
      \frac{1}{m}\rho_{RA_1}+\frac{m-1}{m}\rho_{R}\ox\sigma_{A_1}\Big) \\
   =&\tr\Big[\rho_{RA}\Big(\log\big(\frac{1}{m}\rho_{RA}+\frac{m-1}{m}
      \rho_{R}\ox\sigma_{A}\big)-\log(\rho_{R}\ox\sigma_{A})\Big)\Big],
\end{align*}
where the third line is due to the symmetry of the states
$\frac{1}{m}\sum_{i}\xi_i$ and $\rho_R \ox (\sigma^{\ox m})_{A^m}$ over systems
$A_1,A_2,\cdots,A_m$, and for the inequality we have used the data
processing inequality for relative entropy under partial trace. Now
employ the pinching map $\mc{E}_{\rho_{R}\ox\sigma_{A}}$ and write $\bar{\rho}_{RA}:=\mc{E}_{\rho_{R}\ox\sigma_{A}}(\rho_{RA})$.
The pinching inequality together with the operator monotonicity of the
logarithm gives
\[
\log\Big(\frac{1}{m}\rho_{RA}+\frac{m-1}{m}\rho_{R}\ox\sigma_{A}\Big)\leq
\log\Big(\frac{v}{m}\bar{\rho}_{RA}+\rho_{R}\ox\sigma_{A}\Big).
\]
Making use of this, we proceed as follows.
\[\begin{split}
     & D\Big(\tau_{RA_1A_2\cdots A_m}\big\| \rho_R \ox (\sigma^{\ox m})_{A^m}\Big) \\
\leq & \tr\left[\rho_{RA}\left(\log\big(\frac{v}{m}\bar{\rho}_{RA}+
       \rho_{R}\ox\sigma_{A}\big)-\log(\rho_{R}\ox\sigma_{A})\right)\right] \\
   = & \tr\left[\bar{\rho}_{RA}\log\left(\frac{v}{m}\bar{\rho}_{RA}
       (\rho_{R}\ox\sigma_{A})^{-1}+\1_{RA}\right)\right]              \\
\leq & \frac{1}{(\ln2)s}\tr\Big[\bar{\rho}_{RA} \Big( \frac{v^s}{m^s}
       (\bar{\rho}_{RA})^s (\rho_{R}\ox\sigma_{A})^{-s}\Big)\Big]  \\
   = & \frac{v^s}{(\ln2)s}\exp\Big\{-s\big(\log m-D_{1+s}
       (\bar{\rho}_{RA}\| \rho_R\ox\sigma_A)\big)\Big\}                \\
\leq & \frac{v^s}{(\ln2)s}\exp\Big\{-s\big(\log m-D_{1+s}
       (\rho_{RA}\| \rho_R\ox\sigma_A)\big)\Big\},
\end{split}\]
where for the third line note that the density matrices $\bar{\rho}_{RA}$
and $\rho_{R}\ox\sigma_{A}$ commute, for the fourth line we have used the
inequality $\ln(1+x)\leq \frac{1}{s}x^s$ for $x\geq 0$ and $0<s\leq 1$,
and the last line is by the data processing inequality for the sandwiched
R{\'e}nyi divergence~\cite{WWY2014strong, MDSFT2013on, FrankLieb2013monotonicity,
Beigi2013sandwiched}.
\end{Proof}

\subsection{Smoothing of max-information and conditional min-entropy}
  \label{subsec:MI-CE}
Recall that the smooth max-relative entropy is defined for $\rho \in
\mc{S}(\mc{H})$, $\sigma \in \mc{P}(\mc{H})$, and $0\leq\delta<1$
as~\cite{Datta2009min}
\[
D^\delta_{\rm max}(\rho \| \sigma):=\min\left\{\lambda:\exists\,\tilde{\rho}\in
\mc{S}_{\leq}(\mc{H}) {\rm \ \ s.t.\ \ } P(\tilde{\rho}, \rho) \leq \delta,
\,\tilde{\rho} \leq 2^\lambda\sigma\right\}.
\]
In~\cite{LYH2023tight} we have introduced the smoothing quantity for
the max-relative entropy
\[
\delta(\rho \| \sigma,\lambda):=\min\left\{ P(\tilde{\rho}, \rho): \tilde{\rho}\in
\mc{S}_{\leq}(\mc{H}),\,\tilde{\rho} \leq 2^\lambda\sigma\right\}
\]
and obtained
\begin{equation}
  \label{eq:exp-mre}
\lim_{n\rightarrow\infty} \frac{-1}{n}\log \delta\big(\rho^{\ox n}\big\|
\sigma^{\ox n}, nr\big) = \frac{1}{2}\sup_{s \geq 0}\left\{ s\big(r-D_{1+s}
(\rho \| \sigma)\big)\right\}.
\end{equation}

In this section, we are interested in the smooth max-information and the
smooth conditional min-entropy defined for $\rho_{AB}\in\mc{S}(AB)$ and
$0\leq\delta<1$, respectively as~\cite{Renner2005security,BCR2011the,
ABJT2020partially}
\begin{align*}
   I^\delta_{\rm max}(A:B)_\rho
&:=\ \min_{\sigma_B\in\mc{S}(B)}D^\delta_{\rm max}(\rho_{AB}\|\rho_A\ox\sigma_B),\\
   H^\delta_{\rm min}(A|B)_\rho
&:=-\min_{\sigma_B\in\mc{S}(B)} D^\delta_{\rm max}(\rho_{AB}\|\1_A\ox\sigma_B).
\end{align*}
Our purpose is to derive the asymptotic exponents for the smoothing of these
two one-shot entropies. Note that in the classical case where $\rho_{AB}$
is a probability distribution, this kind of exponential analysis was done
for the conditional min-entropy based on the trace distance~\cite{Hayashi2016security}.
In doing so, we define the corresponding smoothing quantities.
\begin{definition}
  \label{def:smoothing-q}
Let $\rho_{AB}\in\mc{S}(AB)$ and $\lambda\in\mathbb{R}$. The smoothing quantity
for the max-information and for the conditional min-entropy is defined,
respectively, as
\begin{align}
  \label{eq:def-smi}
  \delta_{A:B}(\rho_{AB},\lambda)&:=
\min_{\sigma_B\in\mc{S}(B)} \delta(\rho_{AB}\|\rho_A\ox\sigma_B,\lambda),\\
  \label{eq:def-sce}
  \delta_{A|B}(\rho_{AB},\lambda)&:=
\min_{\sigma_B\in\mc{S}(B)} \delta(\rho_{AB}\|\1_A\ox\sigma_B,-\lambda).
\end{align}
\end{definition}
Note that alternative expressions for the smoothing quantities are
\begin{align}
  \label{eq:alt-smi}
  \delta_{A:B}(\rho_{AB},\lambda)&=
\min\left\{\delta\in\mathbb{R}:I^\delta_{\rm max}(A:B)_\rho\leq\lambda\right\},\\
  \label{eq:alt-sce}
  \delta_{A|B}(\rho_{AB},\lambda)&=
\min\left\{\delta\in\mathbb{R}:H^\delta_{\rm min}(A|B)_\rho\geq\lambda\right\},
\end{align}
from which the relation to the smooth max-information and smooth conditional
min-entropy is easily seen. The main result is stated in the following.
\begin{theorem}
 \label{theorem:exp-mice}
For $\rho_{AB}\in\mc{S}(AB)$ and $r\in\mathbb{R}$, we have
\begin{align}
  \label{eq:exp-mi}
\lim_{n\rar\infty}\frac{-1}{n}\log\,\delta_{A^n:B^n}\left(\rho_{AB}^{\ox n},nr\right)
&=\frac{1}{2}\sup_{s \geq 0}\left\{ s\big(r-I_{1+s}(A:B)_\rho\big)\right\},\\
  \label{eq:exp-ce}
\lim_{n\rar\infty}\frac{-1}{n}\log\,\delta_{A^n|B^n}\left(\rho_{AB}^{\ox n},nr\right)
&=\frac{1}{2}\sup_{s \geq 0}\left\{ s\big(H_{1+s}(A|B)_\rho-r\big)\right\}.
\end{align}
\end{theorem}

\begin{Proof}
Since the r.h.s. of Eqs.~(\ref{eq:def-smi}) and~(\ref{eq:def-sce}) are
similar, it suffices to prove, for any $M\in\mc{P}(A)$ and $r\in\mathbb{R}$,
\beq\label{eq:exp-mices}
\lim_{n\rar\infty}\frac{-1}{n}\log\!\min_{\sigma^{(n)}_{B^n}\in\mc{S}(B^n)}\!\!
\delta\left(\rho_{AB}^{\ox n}\big\|M_A^{\ox n}\ox\sigma^{(n)}_{B^n},nr\right)
\!=\!\frac{1}{2}\sup_{s \geq 0}\left\{s\big(r-\!\!\!\min_{\sigma_B\in\mc{S}(B)}
   \!\!\!D_{1+s}(\rho_{AB}\|M_A\ox\sigma_B)\big)\right\}.
\eeq
Substituting $\rho_A$ for $M_A$ results in Eq.~(\ref{eq:exp-mi}), and the
substitution of $M_A\leftarrow\1$ and $r\leftarrow -r$ recovers
Eq.~(\ref{eq:exp-ce}).

Observing that
\[
     \min_{\sigma^{(n)}_{B^n}\in\mc{S}(B^n)}
\delta\left(\rho_{AB}^{\ox n}\big\|M_A^{\ox n}\ox\sigma^{(n)}_{B^n},nr\right)
\leq \min_{\sigma_B^{\ox n}\in\mc{S}(B^n)}
\delta\left(\rho_{AB}^{\ox n}\big\|M_A^{\ox n}\ox\sigma_B^{\ox n},nr\right),
\]
we obtain the "$\geq$" part of Eq.~(\ref{eq:exp-mices}) by invoking
Eq.~(\ref{eq:exp-mre}). That is,
\beq\label{eq:exp-mices-low}
\liminf_{n\rar\infty}\frac{-1}{n}\log\!\min_{\sigma^{(n)}_{B^n}\in\mc{S}(B^n)}\!\!
\delta\left(\rho_{AB}^{\ox n}\big\|M_A^{\ox n}\ox\sigma^{(n)}_{B^n},nr\right)
\!\geq\!\frac{1}{2}\sup_{s \geq 0}\left\{s\big(r-\!\!\!\min_{\sigma_B\in\mc{S}(B)}
   \!\!\!D_{1+s}(\rho_{AB}\|M_A\ox\sigma_B)\big)\right\}.
\eeq

Now we turn to the proof of the opposite direction. At first, we have
\beq\begin{split}\label{eq:mices-1}
&\min_{\sigma^{(n)}_{B^n}\in\mc{S}(B^n)}\delta\left(\rho_{AB}^{\ox n}
 \big\|M_A^{\ox n}\ox\sigma^{(n)}_{B^n},nr\right)  \\
=&\min\left\{P\big(\rho_{AB}^{\ox n}, \gamma^{(n)}_{A^nB^n}\big):
  \gamma^{(n)}_{A^nB^n}\in\mc{S}_\leq(A^nB^n), \big(\exists\ \sigma^{(n)}_{B^n}\in\mc{S}(B^n)\big)\
  \gamma^{(n)}_{A^nB^n}\leq 2^{nr}M_A^{\ox n}\ox\sigma^{(n)}_{B^n}\right\} \\
=&\min\left\{P\big(\rho_{AB}^{\ox n}, \gamma^{(n)}_{A^nB^n}\big):
  \gamma^{(n)}_{A^nB^n}\in\mc{S}^{\rm sym}_\leq(A^nB^n), \big(\exists\ \sigma^{(n)}_{B^n}
                                                   \in\mc{S}^{\rm sym}(B^n)\big)\
  \gamma^{(n)}_{A^nB^n}\leq 2^{nr}M_A^{\ox n}\ox\sigma^{(n)}_{B^n}\right\} \\
\geq &\min\left\{P\big(\rho_{AB}^{\ox n}, \gamma^{(n)}_{A^nB^n}\big):
  \gamma^{(n)}_{A^nB^n}\in\mc{S}_\leq(A^nB^n),
  \gamma^{(n)}_{A^nB^n}\leq 2^{nr}g_{n, |B|}M_A^{\ox n}\ox\omega^{(n)}
         _{B^n}\right\}, \\
\end{split}\eeq
where the second line is by definition, in the third line we restrict the
minimization to over the symmetric states by making a random permutation
and this makes no difference because random permutation operation as a
CPTP map keeps operator inequality and does not increase the purified
distance, in the fourth line we have employed the universal symmetric
state $\omega^{(n)}_{B^n}$ and made use of Lemma~\ref{lemma:unisym}.
Let $\gamma^{(n)*}_{A^nB^n}$ be the optimal state at which the last
line of Eq.~(\ref{eq:mices-1}) achieves the minimum. Let $\mc{E}^n
\equiv \mc{E}_{M_A^{\ox n}\ox\omega^{(n)}_{B^n}}$ be the pinching map
associated with $M_A^{\ox n}\ox\omega^{(n)}_{B^n}$. Then we have
\beq\begin{split}\label{eq:mices-2}
     \min_{\sigma^{(n)}_{B^n}\in\mc{S}(B^n)}\delta\left(\rho_{AB}^{\ox n}
     \big\|M_A^{\ox n}\ox\sigma^{(n)}_{B^n},nr\right)
\geq &P\left(\rho_{AB}^{\ox n}, \gamma^{(n)*}_{A^nB^n}\right)\\
\geq &P\left(\mc{E}^n\big(\rho_{AB}^{\ox n}\big), \mc{E}^n
       \big(\gamma^{(n)*}_{A^nB^n}\big)\right),
\end{split}\eeq
and in addition, $\gamma^{(n)*}_{A^nB^n}\leq 2^{nr} g_{n, |B|}M_A^{\ox n}
\ox\omega^{(n)}_{B^n}$, which after the CPTP map $\mc{E}^n$ being applied
to both sides yields
\beq\label{eq:mices-3}
\mc{E}^n\big(\gamma^{(n)*}_{A^nB^n}\big)\leq 2^{nr} g_{n, |B|}M_A^{\ox n}
\ox\omega^{(n)}_{B^n}.
\eeq

To proceed, we construct a projective measurement $\{\Pi_n,\1-\Pi_n\}$
with
\[
\Pi_n:=\Big\{\mc{E}^n\big(\rho_{AB}^{\ox n}\big)\geq 9\cdot2^{nr}
        g_{n, |B|}M_A^{\ox n}\ox\omega^{(n)}_{B^n}\Big\}
\]
and set
\begin{align}
  \label{eq:p_n}
p_n=&\tr\big(\mc{E}^n(\rho_{AB}^{\ox n})\Pi_n\big), \\
  \label{eq:q_n}
q_n=&\tr\big(\mc{E}^n(\gamma^{(n)*}_{A^nB^n})\Pi_n\big).
\end{align}
By Eq.~(\ref{eq:mices-3}) and the construction of $\Pi_n$, it can be verified
that
\[q_n \leq \frac{1}{9} p_n.\]
So letting $\mc{M}^n:X\mapsto\tr(X\Pi_n)\proj{0}+\tr(X(\1-\Pi_n))\proj{1}$
be the measurement map associated with $\{\Pi_n,\1-\Pi_n\}$, we are able
to obtain
\beq\begin{split}\label{eq:mices-4}
     &\min_{\sigma^{(n)}_{B^n}\in\mc{S}(B^n)}\delta\left(\rho_{AB}^{\ox n}
       \big\|M_A^{\ox n}\ox\sigma^{(n)}_{B^n},nr\right)               \\
\geq &P\left(\mc{M}^n\circ\mc{E}^n\big(\rho_{AB}^{\ox n}\big),
       \mc{M}^n\circ\mc{E}^n\big(\gamma^{(n)*}_{A^nB^n}\big)\right)   \\
  =  &\sqrt{1-\left(\sqrt{p_n}\sqrt{q_n}+\sqrt{1-p_n}
       \sqrt{\tr(\gamma^{(n)*}_{A^nB^n})-q_n}\right)^2}               \\
\geq &\sqrt{p_n\big(\frac{1}{3}-\frac{p_n}{9}\big)}    \\
\geq &\frac{\sqrt{2p_n}}{3},
\end{split}\eeq
where the second line follows from Eq.~(\ref{eq:mices-2}) and the data
processing inequality for purified distance, for the fourth line we
have used $q_n \leq \frac{1}{9} p_n$ and $\tr(\gamma^{(n)*}_{A^nB^n})
-q_n\leq1$, and for the last line note that $p_n\leq 1$.

We prove later in Proposition~\ref{prop:usym-Imax} the asymptotics for
$p_n$ (cf. Eq.~(\ref{eq:p_n}) for its expression), that is,
\beq\label{eq:asym-p}
 \lim_{n\rar\infty}\frac{-1}{n}\log p_n
=\sup_{s \geq 0}\left\{s\big(r-\min_{\sigma_B\in\mc{S}(B)}
                    D_{1+s}(\rho_{AB}\|M_A\ox\sigma_B)\big)\right\}.
\eeq
Then it follows from Eqs.~(\ref{eq:mices-4}) and (\ref{eq:asym-p})
that
\beq\label{eq:exp-mices-up}
\limsup_{n\rar\infty}\frac{-1}{n}\log\!\min_{\sigma^{(n)}_{B^n}\in
\mc{S}(B^n)}\!\!\delta\left(\rho_{AB}^{\ox n}\big\|M_A^{\ox n}\ox
\sigma^{(n)}_{B^n},nr\right)\!\leq\!\frac{1}{2}\sup_{s \geq 0}\left\{
s\big(r-\!\!\!\min_{\sigma_B\in\mc{S}(B)}\!\!\!D_{1+s}(\rho_{AB}\|M_A
\ox\sigma_B)\big)\right\}.
\eeq

Eventually, Eqs.~(\ref{eq:exp-mices-up}) and (\ref{eq:exp-mices-low})
together lead to Eq.~(\ref{eq:exp-mices}), and we are done.
\end{Proof}

\begin{proposition}
  \label{prop:usym-Imax}
Let $\rho_{AB}\in\mc{S}(AB)$ and $M_A\in\mc{P}(A)$ be such that $\supp(\rho_A)
\subseteq\supp(M_A)$. Let $\omega^{(n)}_{B^n}$ be the universal symmetric
state. Let $\mc{E}^n\equiv \mc{E}_{M_A^{\ox n}\ox\omega^{(n)}_{B^n}}$ be
the pinching map, and let $f(n)\geq 1$ be any sub-exponential function of $n$.
For given $r\in\mathbb{R}$, consider the sequence
\[
p_n:=\tr\left[\mc{E}^n\big(\rho_{AB}^{\ox n}\big)
     \Big\{\mc{E}^n\big(\rho_{AB}^{\ox n}\big)\geq f(n)2^{nr}
     M_A^{\ox n}\ox\omega^{(n)}_{B^n}\Big\} \right]
\]
for $n\in\mathbb{N}$. We have
\beq\label{eq:usym-Imax}
 \lim_{n\rar\infty}\frac{-1}{n}\log p_n
=\sup_{s \geq 0}\left\{s\big(r-\min_{\sigma_B\in\mc{S}(B)}
 D_{1+s}(\rho_{AB}\|M_A\ox\sigma_B)\big)\right\}.
\eeq
\end{proposition}

In order to prove Proposition~\ref{prop:usym-Imax}, we employ two technical
results from the literature. The following lemma is due to Hayashi and
Tomamichel~\cite{HayashiTomamichel2016correlation}.
\begin{lemma}
  \label{lem:additive-Is}
Let $\rho_{AB}$, $M_A$, $\omega^{(n)}_{B^n}$ and $\mc{E}^n$ be given in
Proposition~\ref{prop:usym-Imax}. For any $s\geq \frac{1}{2}$ we have
\[
\frac{1}{n}D_s\Big(\mc{E}^n(\rho_{AB}^{\ox n})\big\|M_A^{\ox n}\ox
\omega^{(n)}_{B^n}\Big) = \min_{\sigma_B\in\mc{S}(B)}D_s(\rho_{AB}\|M_A
\ox\sigma_B)+O\Big(\frac{\log n}{n}\Big).
\]
The implicit constants of the last term are independent of $s$.
\end{lemma}

\bigskip
The following lemma is a variant of the G\"{a}rtner-Ellis theorem of large
deviation theory. It is due to Chen~\cite{Chen2000generalization} and
reformulated in~\cite{HayashiTomamichel2016correlation}.
\begin{lemma}
  \label{lem:Gartner-Ellis}
Let $\{Z_n\}_{n\in\mathbb{N}}=:\mc{Z}$ be a sequence of random variables,
with the asymptotic cumulant generating function
\[
\Lambda_{\mc{Z}}(t)
:=\lim_{n\rar\infty}\frac{1}{n}\log\mathbb{E}\big[2^{ntZ_n}\big].
\]
If the above limit exists and the function $t\mapsto\Lambda_{\mc{Z}}(t)$ is
differentiable in some interval $(a,b)$, then for any
$z\in\big(\lim\limits_{t\searrow a}\Lambda_{\mc{Z}}'(t),\lim\limits_
{t\nearrow b}\Lambda_{\mc{Z}}'(t)\big)$, it holds that
\[
\limsup_{n\rar\infty}\frac{-1}{n}\log\,{\rm Pr}\{Z_n\geq z\}
\leq \sup_{t\in(a,b)}\big\{zt-\Lambda_{\mc{Z}}(t)\big\}.
\]
\end{lemma}

With Lemmas~\ref{lem:additive-Is} and \ref{lem:Gartner-Ellis}, we are now
ready for the proof of Proposition~\ref{prop:usym-Imax}.

\medskip
\begin{Proofof}[Proposition~\ref{prop:usym-Imax}]
For simplicity, we use the shorthands
\[
P_n:=\mc{E}^n\big(\rho_{AB}^{\ox n}\big),\quad
Q_n:=M_A^{\ox n}\ox\omega^{(n)}_{B^n}.
\]
We will make frequent use of the facts that $P_n$ commutes with $Q_n$ as
well as that, by Lemma~\ref{lem:additive-Is},
\beq\label{eq:Ds-asympt}
\frac{1}{n}D_{1+s}(P_n\|Q_n)\longrightarrow \min_{\sigma_B\in\mc{S}(B)}
 D_{1+s}(\rho_{AB}\|M_A\ox\sigma_B), \quad \text{as}\ \ n\rar\infty.
\eeq
In what follows we will first show that the r.h.s. of Eq.~(\ref{eq:usym-Imax})
is an achievable rate for the exponential decreasing of $p_n$, and then
we prove that this is indeed the optimal rate.

Now we prove the achievability part. For any $s>0$,
\[\begin{split}
     \frac{-1}{n}\log p_n
  &= \frac{-1}{n}\log\tr\big[P_n\left\{P_n\geq f(n) 2^{nr}Q_n\right\}\big]\\
  &\geq\frac{-1}{n}\log\tr\Big[P_n\Big(\frac{P_n}{f(n) 2^{nr}Q_n}\Big)^s\Big]\\
  & = sr-\frac{1}{n}\log\tr\big[P_n^{1+s}Q_n^{-s}\big]+\frac{s\log f(n)}{n}\\
  & = sr-\frac{s}{n}D_{1+s}\big(P_n\|Q_n\big)+\frac{s\log f(n)}{n}\\
  &\overset{n\to\infty}{\longrightarrow} s\big(r-\min_{\sigma_B\in\mc{S}(B)}D_{1+s}(\rho_{AB}\|M_A\ox\sigma_B)\big).
\end{split}\]
Noticing that $s>0$ is arbitrary, we conclude from the above  estimation
that
\beq\label{eq:usym-Imax-low}
    \liminf_{n\rar\infty}\frac{-1}{n}\log p_n
\geq\sup_{s \geq 0}\left\{s\big(r-\min_{\sigma_B\in\mc{S}(B)}
                         D_{1+s}(\rho_{AB}\|M_A\ox\sigma_B)\big)\right\}.
\eeq

For the other direction, we follow the method
of~\cite{HayashiTomamichel2016correlation} and employ the G\"{a}rtner-Ellis
theorem. Let
\[
Z_n:=\frac{1}{n}\big(\log P_n-\log Q_n -nr -\log f(n) \big).
\]
We have
\beq\begin{split}\label{eq:usym-Imax-1}
p_n=&\tr\big[P_n\left\{P_n\geq f(n) 2^{nr}Q_n\right\}\big]  \\
   =&\tr\big[P_n\left\{Z_n\geq 0\right\}\big].
\end{split}\eeq
We see that the observable $Z_n$ commutes with the state $P_n$. So the
above Eq.~(\ref{eq:usym-Imax-1}) has a classical-probability-theoretic
explanation in which $Z_n$ is regarded as a random variable with $P_n$
being its distribution. Specifically, let the set of orthonormal vectors
$\{\ket{a_x^n}\}_x$ be the common eigenvectors of $Z_n$ and $P_n$. Then
$Z_n$ takes value $\bra{a^n_x}Z_n\ket{a^n_x}$ with probability
$\bra{a^n_x}P_n\ket{a^n_x}$. Now, Eq.~\eqref{eq:usym-Imax-1} translates
to
\[
p_n=\operatorname{Pr}\big\{Z_n\geq 0\big\},
\]
with respect to the probability distribution $\{\bra{a^n_x}P_n\ket{a^n_x}
\}_x$. To apply the G\"{a}rtner-Ellis theorem (in the form of
Lemma~\ref{lem:Gartner-Ellis}), we calculate the asymptotic cumulant
generating function of the sequence $\mc{Z}=\{Z_n\}_{n\in\mathbb{N}}$:
\beq\begin{split}\label{eq:cgf}
  \Lambda_{\mc{Z}}(s)
=&\lim_{n\rar\infty}\frac{1}{n}\log\mathbb{E}\big[2^{nsZ_n}\big]\\
=&\lim_{n\rar\infty}\frac{1}{n}\log\mathbb{E}
     \Big[\big(\frac{P_n}{Q_n}\big)^s(f(n))^{-s}2^{-nsr}\Big]\\
=&\lim_{n\rar\infty}\frac{1}{n}\log\tr
     \Big[P_n\big(\frac{P_n}{Q_n}\big)^s(f(n))^{-s}2^{-nsr}\Big]\\
=&\lim_{n\rar\infty}\Big(\frac{s}{n}D_{1+s}(P_n\|Q_n)-sr-
                                      \frac{s\log f(n)}{n}\Big) \\
=&\ s\big(\min_{\sigma_B\in\mc{S}(B)}
                         D_{1+s}(\rho_{AB}\|M_A\ox\sigma_B)-r\big).
\end{split}\eeq
Set $F(s):=\min\limits_{\sigma_B\in\mc{S}(B)}D_{1+s}(\rho_{AB}\|M_A\ox
\sigma_B)$ for later convenience. To proceed, we restrict our attention
to the case where $F(0)<r<F(+\infty)$ at the moment. It has been proven
in~\cite{HayashiTomamichel2016correlation} that the function $s\mapsto F(s)$
is monotonically increasing, continuously
differentiable, and $s\mapsto sF(s)$ is convex. As a result, we have that
$s\mapsto \Lambda_{\mc{Z}}(s)$ is differentialbe in $(0,+\infty)$, and
\begin{align*}
      \lim_{s\rar 0} \Lambda_{\mc{Z}}'(s)
&  =  D(\rho_{AB}\|M_A\ox\rho_B)-r=F(0)-r<0, \\
      (\exists s_0>0)\ \lim_{s\rar s_0} \Lambda_{\mc{Z}}'(s)
&\geq \min_{\sigma_B\in\mc{S}(B)}D_{1+s_0}(\rho_{AB}\|M_A\ox\sigma_B)-r
      =F(s_0)-r>0.
\end{align*}
Hence, Lemma~\ref{lem:Gartner-Ellis} applies, yielding for $r\in(F(0),
F(\infty))$,
\beq\begin{split}\label{eq:usym-Imax-up}
     \limsup_{n\rar\infty}\frac{-1}{n}\log p_n
   =&\limsup_{n\rar\infty}\frac{-1}{n}\log{\rm Pr}\big\{Z_n\geq 0\big\} \\
\leq&\!\sup_{0<s<s_0}\!\left\{s\big(r-\min_{\sigma_B\in\mc{S}(B)}
                         D_{1+s}(\rho_{AB}\|M_A\ox\sigma_B)\big)\right\}\\
\leq&\ \ \sup_{s>0}\ \ \left\{s\big(r-\min_{\sigma_B\in\mc{S}(B)}
                         D_{1+s}(\rho_{AB}\|M_A\ox\sigma_B)\big)\right\}.
\end{split}\eeq

Now combining Eqs.~(\ref{eq:usym-Imax-low}) and (\ref{eq:usym-Imax-up})
together lets us arrive at
\beq\label{eq:usym-Imax-done}
\lim_{n\rar\infty}\frac{-1}{n}\log p_n
=\sup_{s \geq 0}\left\{s\big(r-\min_{\sigma_B\in\mc{S}(B)}
 D_{1+s}(\rho_{AB}\|M_A\ox\sigma_B)\big)\right\}
\eeq
for $r\in(F(0),F(\infty))$. To complete the proof, we show that the
equality of Eq.~(\ref{eq:usym-Imax-done}) can be extended to the whole
range $r\in\mathbb{R}$. At first, we observe by definition that, the l.h.s.
of Eq.~(\ref{eq:usym-Imax-done}) is nonnegative and monotonically
increasing with $r$. Moreover, the equality of Eq.~(\ref{eq:usym-Imax-done})
established for $r\in(F(0),F(\infty))$ shows that the l.h.s. goes
to $0$ when $r\searrow F(0)=D(\rho_{AB}\|M_A\ox\rho_B)$. So we conclude
that the l.h.s. of Eq.~(\ref{eq:usym-Imax-done}) equals $0$ for $r\leq F(0)$,
coinciding with the right hand side. Next, we consider the case $r>F(\infty)
=\min_{\sigma_B\in\mc{S}(B)}D_{\rm max}(\rho_{AB}\|M_A\ox\sigma_B)$.
Letting $s\rar\infty$ in Lemma~\ref{lem:additive-Is} we get
\[
D_{\rm max}\Big(\mc{E}^n(\rho_{AB}^{\ox n})\big\|M_A^{\ox n}\ox
\omega^{(n)}_{B^n}\Big) =n\left( F(\infty)+O\Big(\frac{\log n}{n}\Big)\right).
\]
This implies that, for $r>F(\infty)$ and $n$ big enough,
\[
\Big\{\mc{E}^n\big(\rho_{AB}^{\ox n}\big)\geq
f(n)2^{nr}M_A^{\ox n}\ox\omega^{(n)}_{B^n}\Big\}=0
\]
and hence $p_n=0$. As a result, we conclude that the l.h.s. of
Eq.~(\ref{eq:usym-Imax-done}) is $+\infty$ for $r>F(\infty)$,
coinciding with the right hand side, too.
\end{Proofof}

\subsection{Proof of Theorem~\ref{thm:reliability-dec}}
\label{subsec:thm-proof}
Based on the results obtained in the above two subsections, we are now ready
to complete the proof of Theorem~\ref{thm:reliability-dec}. At first, we
prove the one-shot achievability bound of Theorem~\ref{thm:dec-p-up}.

\begingroup
\def\thetheorem{\ref{thm:dec-p-up}}
\begin{theorem}[restatement]
Let $\rho_{RA}\in \mc{S}(RA)$. For any $m\in\mathbb{N}$, $0<s\leq 1$ and
$\sigma_A\in\mc{S}(A)$, the optimal performance of decoupling $A$ from $R$
is bounded as
\[
P^{\rm dec}_{R:A}(\rho_{RA},\log m) \leq
\sqrt{\frac{v^s}{s}}\exp\Big\{-s\big(\log m-\frac{1}{2}D_{1+s}
     (\rho_{RA}\| \rho_R\ox\sigma_A)\big)\Big\},
\]
where $v$ is the number of distinct eigenvalues of $\rho_R\ox \sigma_A$.
\end{theorem}
\addtocounter{theorem}{-1}
\endgroup

\begin{Proof}
We will employ the convex-split lemma. The convex-split lemma provides
directly a catalytic decoupling strategy via random unitary operation,
which in turn can be converted into a way for catalytic decoupling via
removing a subsystem (cf. Section~\ref{subsec:relation-dec}). Identify
$A$ with $A_1$, and let $A'=A_2A_3\cdots A_m$, where all the $A_i$ systems
have equal dimension. Let $A'$ in the state $\bar{\sigma}_{A'}=\sigma_{A_2}
\ox\sigma_{A_3}\ox\cdots\ox\sigma_{A_m}$ be the catalytic system. We
construct a random unitary channel
\[
\Lambda_{AA'}: X \mapsto \frac{1}{m}\sum_{i=1}^mU_iXU_i^*,
\]
where $U_i=W_{(1,i)}$ is the swapping between $A_1$ and $A_i$ (we set
$W_{(1,1)}=\1$). Then
\[
\Lambda_{AA'}\big(\rho_{RA}\ox \bar{\sigma}_{A'}\big)
=\frac{1}{m}\sum_{i=1}^m \rho_{RA_i}\ox\big[\sigma^{\ox m-1}\big]_{A^m/A_i},
\]
which is in the form of the state in Lemma~\ref{lem:convex-split}. So we
have
\[\begin{split}
   P_{R:A}^{\rm dec\text{-}u}(\rho_{RA}, \log m)
\leq &P\Big(\Lambda_{AA'}\big(\rho_{RA}\ox \bar{\sigma}_{A'}\big),
                  \rho_{R}\ox\big[\sigma^{\ox m}\big]_{A^m}\Big) \\
\leq &\Big[(\ln2)\ D\Big(\Lambda_{AA'}\big(\rho_{RA}\ox \bar{\sigma}_{A'}\big)
      \big\|\rho_R\ox\big[\sigma^{\ox m}\big]_{A^m} \Big)\Big]^{\frac{1}{2}}\\
\leq &\sqrt{\frac{v^s}{s}}\exp\Big\{-s\big(\frac{1}{2}\log m-
       \frac{1}{2}D_{1+s}(\rho_{RA}\| \rho_R\ox\sigma_A)\big)\Big\},
\end{split}\]
where for the second line we have used the relation $P(\rho,\sigma)\leq
\sqrt{(\ln 2)D(\rho\|\sigma)}$, and for the third line we have used
Lemma~\ref{lem:convex-split}. At last, invoking
Proposition~\ref{prop:dec-relation} lets us confirm the claim.
\end{Proof}

\medskip
The following one-shot converse bound is proved using standard techniques
(see, e.g.~\cite{DBWR2014one} and~\cite{MBDRC2017catalytic}).
\begin{proposition}
  \label{prop:dec-p-low}
Let $\rho_{RA}\in \mc{S}(RA)$. For $k\geq 0$, the optimal performance of
decoupling $A$ from $R$ is bounded by the smoothing quantity for the
max-information (cf. Eq.~(\ref{eq:def-smi}) in Definition~\ref{def:smoothing-q}):
\[
P^{\rm dec}_{R:A}(\rho_{RA},k) \geq \delta_{R:A}(\rho_{RA},2k).
\]
\end{proposition}

\begin{Proof}
Consider an arbitrary decoupling scheme that discards not more than than $k$
qubits. Let the catalytic state be $\sigma_{A'}$ and the unitary decoupling
operation be $U:\mc{H}_{AA'}\rar\mc{H}_{A_1A_2}$. Then $\log |A_2|\leq k$
and the performance for this scheme is
\[
\epsilon:=\min_{\omega_{A_1}\in\mc{S}(A_1)} P\big(\tr_{A_2}U(\rho_{RA}\ox
\sigma_{A'})U^*,\rho_R \ox\omega_{A_1}\big).
\]
Let $\omega_{A_1}^*$ be the optimal state in the above minimization. By
Uhlmann's theorem~\cite{Uhlmann1976transition}, there is a state
$\tilde{\rho}_{RA_1A_2}$ such that
\begin{align}
  \label{eq:decpl-1}
&P\big(U(\rho_{RA}\ox\sigma_{A'})U^*,\tilde{\rho}_{RA_1A_2}\big)=\epsilon,\\
  \label{eq:decpl-2}
&\tr_{A_2} \tilde{\rho}_{RA_1A_2} = \rho_R \ox\omega_{A_1}^*.
\end{align}
From Eq.~(\ref{eq:decpl-1}) we get
\beq\label{eq:decpl-3}
P\big(\rho_{RA},\tr_{A'}U^*\tilde{\rho}_{RA_1A_2}U\big)\leq\epsilon.
\eeq
Eq.~(\ref{eq:decpl-2}) implies that
\[
\tilde{\rho}_{RA_1A_2}
\leq|A_2|^2\rho_R\ox\omega_{A_1}^*\ox \frac{\1_{A_2}}{|A_2|}
\leq 2^{2k}\rho_R\ox\omega_{A_1}^*\ox\frac{\1_{A_2}}{|A_2|},
\]
which further yields
\beq\label{eq:decpl-4}
\tr_{A'}U^*\tilde{\rho}_{RA_1A_2}U\leq2^{2k}\rho_R\ox
\tr_{A'}U^*\big(\omega_{A_1}^*\ox\frac{\1_{A_2}}{|A_2|}\big)U.
\eeq
Inspecting the definition in Eq.~(\ref{eq:def-smi}), we obtain from
Eqs.~(\ref{eq:decpl-3}) and (\ref{eq:decpl-4}) that
\[\delta_{R:A}(\rho_{RA},2k)\leq\epsilon.\]
Since by assumption $\epsilon$ is the performance of an arbitrary decoupling
scheme with cost not larger than $k$, the optimal performance satisfies the
above relation as well. That is,
\[\delta_{R:A}(\rho_{RA},2k)\leq P^{\rm dec}_{R:A}(\rho_{RA},k).\]
\end{Proof}

\begin{Proofof}[Theorem~\ref{thm:reliability-dec}]
By applying Theorem~\ref{thm:dec-p-up} and making the substitution
$\rho_{RA}\leftarrow\rho_{RA}^{\ox n}$, $\log m\leftarrow nr$, $\sigma_A
\leftarrow\sigma_A^{\ox n}$, we get
\[
\frac{-1}{n}\log P^{\rm dec}_{R^n:A^n}\big(\rho_{RA}^{\ox n},nr\big)\geq
\frac{-1}{2n}\log\frac{v_n^s}{s}+s\big(r-\frac{1}{2}D_{1+s}(\rho_{RA}\|
\rho_R\ox\sigma_A)\big),
\]
where $v_n\leq(n+1)^{|R| + |A|}$ is the number of distinct eigenvalues
of $\rho_R^{\ox n}\ox\sigma_A^{\ox n}$. Letting $n\rar\infty$, and then
maximizing the r.h.s. over $s\in (0,1]$ and $\sigma_A\in\mc{S}(A)$, we
arrive at
\[
E^{\rm dec}_{R:A}(\rho_{RA},r)
\geq\liminf_{n\rar\infty}\frac{-1}{n}\log
   P^{\rm dec}_{R^n:A^n}\big(\rho_{RA}^{\ox n},nr\big)
\geq\sup_{0<s\leq 1}\Big\{s\big(r-\frac{1}{2}I_{1+s}(R:A)_\rho\big)\Big\},
\]
which is Eq.~(\ref{eq:exp-dec-low}).
For Eq.~(\ref{eq:exp-dec-up}), we apply Proposition~\ref{prop:dec-p-low}
and Theorem~\ref{theorem:exp-mice} to obtain
\[\begin{split}
      E^{\rm dec}_{R:A}(\rho_{RA},r)
  = & \limsup_{n\rar\infty}\frac{-1}{n}\log
      P^{\rm dec}_{R^n:A^n}\big(\rho_{RA}^{\ox n},nr\big) \\
\leq& \lim_{n\rar\infty}\frac{-1}{n}\log
      \delta_{R^n:A^n}\big(\rho_{RA}^{\ox n},2nr\big)   \\
  = & \sup_{s\geq0}\Big\{s\big(r-\frac{1}{2}I_{1+s}(R:A)_\rho\big)\Big\}.
\end{split}\]
At last, Eq.~(\ref{eq:exp-dec}) follows from Eqs.~(\ref{eq:exp-dec-low}) and
(\ref{eq:exp-dec-up}). To see this, we consider the function
\beq
f(s)=s\big(r-\frac{1}{2}I_{1+s}(R:A)_\rho\big)
\eeq
for $s\in[0,\infty)$. $f(s)$ is concave because $s\mapsto sI_{1+s}(R:A)_\rho$
is convex. So, we have that $f(s)$ reaches the maximum in $[0,1]$ and
therefore
\beq
\max_{s\geq 0}f(s) = \max_{0\leq s\leq 1} f(s),
\eeq
if $f'(1)\leq 0$. This condition is equivalent to $r\leq R_{\rm critical}\equiv
\frac{1}{2}\frac{\mathrm{d}}{\mathrm{d}s} sI_{1+s}(R:A)_{\rho}\big|_{s=1}$.
\end{Proofof}

\begin{remark} \label{rem:mainresult}
Some comments on the main results are as follows.
\begin{enumerate}[(a)]
\item Although the reliability function $E^{\rm dec}_{R:A}(\rho_{RA},r)$ is
      defined as a limit superior, we see in the above proof that when we are
      able to determine it for $r\leq R_{\rm critical}$, the limit of $\frac{-1}
      {n}\log P^{\rm dec}_{R^n:A^n}\big(\rho_{RA}^{\ox n},nr\big)$ actually
      exists.
\item When we employ the convex-split lemma (Lemma~\ref{lem:convex-split})
      to derive the achievability results of Theorem~\ref{thm:dec-p-up} and
      Eq.~(\ref{eq:exp-dec-low}), we can keep using the relative entropy as the
      measure of closeness between states. Meanwhile, Eq.~(\ref{eq:exp-dec-up})
      can be converted into an upper bound for the rate of exponential decay
      of the performance measured by relative entropy, making use of the
      relation $P(\rho,\sigma)\leq\sqrt{(\ln 2)D(\rho\|\sigma)}$. This actually
      enables us to derive results similar to Theorem~\ref{thm:reliability-dec},
      with the purified distance being replaced by the relative entropy in the
      definition of the reliability function.
\item The upper bound of Eq.~\eqref{eq:exp-dec-up} diverges when $r>\frac{1}{2}
      I_{\rm max}(R:A)_\rho$, which can be easily seen. One may guess
      that it goes to $\infty$ when $r\nearrow\frac{1}{2}I_{\rm max}(R:A)_\rho$.
      However, we show in Proposition~\ref{prop:bounded} in the Appendix
      that it is actually bounded when $r<\frac{1}{2}I_{\rm max}(R:A)_\rho$.
\end{enumerate}
\end{remark}

\section{Proof of the relations between quantum tasks}
  \label{sec:relations}
In this section, we prove Proposition~\ref{prop:dec-relation} on
the relation between the three types of catalytic quantum information
decoupling, and Proposition~\ref{prop:merging-dec} on the relation
between catalytic quantum information decoupling and quantum state
merging.

\subsection{Proof of the relation between decouplings}
  \label{subsec:relation-dec}
\begingroup
\def\thetheorem{\ref{prop:dec-relation}}
\begin{proposition}[restatement]
For $\rho_{RA} \in \mathcal{S}(RA)$ and $k\geq 0$, $r \geq 0$, we have
\begin{align}
&P_{R:A}^{\rm dec}(\rho_{RA}, k)
=P_{R:A}^{\rm dec\text{-}m}(\rho_{RA}, 2k)
=P_{R:A}^{\rm dec\text{-}u}(\rho_{RA}, 2k), \label{eqr:dec-P} \\
&E_{R:A}^{\rm dec}(\rho_{RA}, r)
=E_{R:A}^{\rm dec\text{-}m}(\rho_{RA}, 2r)
=E_{R:A}^{\rm dec\text{-}u}(\rho_{RA}, 2r). \label{eqr:dec-E}
\end{align}
\end{proposition}
\addtocounter{theorem}{-1}
\endgroup

\begin{Proof}
Eq.~\eqref{eqr:dec-E} follows from Eq.~\eqref{eqr:dec-P} directly. So it
suffices to prove Eq.~\eqref{eqr:dec-P}. In the following, we show the
identities $P_{R:A}^{\rm dec}(\rho_{RA}, k)=P_{R:A}^{\rm dec\text{-}m}
(\rho_{RA}, 2k)$ and $P_{R:A}^{\rm dec}(\rho_{RA}, k)
=P_{R:A}^{\rm dec\text{-}u}(\rho_{RA}, 2k)$ separately.

We first prove the former identity. Consider an arbitrary decoupling scheme
via removing a subsystem of size not more than $k$ qubits. Let $\sigma_{A'}$
be the catalytic state and $U:\mc{H}_{AA'}\rar \mc{H}_{A_1A_2}$ be the
unitary operation, where $A_2$ is the system to be removed and $\log |A_2|
\leq k$. Introduce a system $A''$ in the maximally mixed state
$\frac{\1_{A''}}{|A''|}$ such that $|A''|=|A_2|$. We construct a decoupling
scheme via projective measurement: the catalytic state is chosen to be
$\sigma_{A'}\ox \frac{\1_{A''}}{|A''|}$, and the projective measurement
is given by $\{Q^x_{AA'A''}\}_{x=1}^m$, with $m=|A_2|^2$ and
\[
  Q^x_{AA'A''}=(U\ox\1_{A''})^*\big((W_{A_2}^x\ox\1_{A''})\Psi_{A_2A''}
               (W_{A_2}^x\ox\1_{A''})^*\ox\1_{A_1}\big)(U\ox\1_{A''}),
\]
where $\{W_{A_2}^x\}_{x=1}^{m}$ is the set of discrete Weyl operators on
$\mc{H}_{A_2}$ and $\Psi_{A_2A''}$ is the maximally entangled state. So,
the cost is $\log m=\log |A_2|^2 \leq 2k$. By the definition of
$P_{R:A}^{\rm{dec\text{-}m}}$, for arbitrary state $\omega_{A_1}$, we have
\begin{equation}
\label{eq:decmleqdec}
\begin{split}
     &P_{R:A}^{\rm{dec\text{-}m}}(\rho_{RA}, 2k) \\
\leq &P\Big(\sum_{x=1}^m Q^x_{AA'A''}\big(\rho_{RA} \ox \sigma_{A'} \ox
       \frac{\1_{A''}}{|A''|}\big)Q^x_{AA'A''},\,\rho_R\ox U^*\big(\omega_{A_1}
       \ox \frac{\1_{A_2}}{|A_2|}\big)U \ox \frac{\1_{A''}}{|A''|}\Big) \\
 =   &P\Big(U^*\big(\tr_{A_2} \big[U(\rho_{RA} \ox \sigma_{A'})U^*\big]
       \ox \frac{\1_{A_2}}{|A_2|}\big)U\ox \frac{\1_{A''}}{|A''|},\,
       \rho_R\ox U^*\big(\omega_{A_1}\ox \frac{\1_{A_2}}{|A_2|}\big)U
       \ox \frac{\1_{A''}}{|A''|}\Big) \\
 =   &P\big(\tr_{A_2}U(\rho_{RA}\ox\sigma_{A'})U^*,\rho_R\ox\omega_{A_1}\big),
\end{split}
\end{equation}
where the third line can be verified by direct calculation. Minimizing
the last line of Eq.~(\ref{eq:decmleqdec}) over $\sigma_{A'}$,
$\omega_{A_1}$ and $U:\mc{H}_{AA'}\rar \mc{H}_{A_1A_2}$, and by
Definition~\ref{def:perform-dec}, we get from Eq.~(\ref{eq:decmleqdec}) that
\begin{equation}
\label{eq:decmleqdec1}
P_{R:A}^{\rm{dec\text{-}m}}(\rho_{RA},2k)\leq P^{\rm{dec}}_{R:A}(\rho_{RA},k).
\end{equation}
Next, we deal with the other direction. Let $\sigma_{A'}$ and $\{Q^x_{AA'}\}
_{x=1}^m$ be the catalytic state and the projectors of a decoupling scheme
via projective measurement with cost $\log m\leq 2k$. Introduce systems $C$
and $C'$ such that $|C|=|C'|=\sqrt{m}$ and denote by $\Psi_{CC'}$ the maximally
entangled state on $C$ and $C'$. We construct a decoupling scheme via removing
a subsystem: the catalytic state is given by $\sigma_{A'}\ox \frac{\1_{CC'}}
{|C|\cdot|C'|}$, and the unitary operation is given by
\[
U_{AA'CC'}
=\sum_{x=1}^m \Big(\sum_{y=1}^m \mathrm{e}^{\frac{2\pi\mathrm{i}xy}{m}}
 Q_{AA'}^y\Big)\ox \big(W^x_C\ox\1_{C'}\big)\Psi_{CC'}
 \big(W^x_C\ox\1_{C'}\big)^*,
\]
where $\{W^x_C\}_{x=1}^m$ is the set of discrete Weyl operators on
$\mc{H}_C$. Set $A_1\equiv AA'C$ be the remaining system, and $A_2\equiv C'$
be the removed system. So, the cost is $\log|A_2|=\log|C'|\leq k$.
By the definition of $P_{R:A}^{\rm{dec}}$, for arbitrary state
$\omega_{AA'}$, we have
\begin{equation}
\label{eq:decleqdecm}
\begin{split}
     & P_{R:A}^{\rm{dec}}(\rho_{RA}, k)\\
\leq & P\Big(\tr_{C'}\big[U_{AA'CC'}\big(\rho_{RA}\ox\sigma_{A'}\ox
        \frac{\1_{CC'}}{|C|\cdot|C'|}\big)U_{AA'CC'}^*\big],\,
        \rho_R\ox\omega_{AA'}\ox\frac{\1_C}{|C|} \Big) \\
  =  & P\Big(\sum_xQ^x_{AA'}\big(\rho_{RA}\ox \sigma_{A'}\big)Q^x_{AA'}\ox
        \frac{\1_C}{|C|},\,\rho_R \ox \omega_{AA'} \ox \frac{\1_C}{|C|} \Big)\\
  =  & P\Big(\sum_xQ^x_{AA'}\big(\rho_{RA}\ox \sigma_{A'}\big)Q^x_{AA'},\,
        \rho_R\ox\omega_{AA'}\Big),
\end{split}
\end{equation}
where the third line can be verified by direct calculation. Minimizing the
last line of Eq.~(\ref{eq:decleqdecm}) over $\sigma_{A'}$, $\omega_{AA'}$
and $\{Q^x_{AA'}\}_{x=1}^m$, and by Definition~\ref{def:perform-dec-m},
we get from Eq.~(\ref{eq:decleqdecm}) that
\begin{equation}
  \label{eq:decleqdecm1}
P_{R:A}^{\rm{dec}}(\rho_{RA},k) \leq P_{R:A}^{\rm{dec\text{-}m}}(\rho_{RA},2k).
\end{equation}
The combination of Eqs.~\eqref{eq:decmleqdec1} and \eqref{eq:decleqdecm1}
leads to the identity $P_{R:A}^{\rm dec}(\rho_{RA}, k)=P_{R:A}^{\rm dec\text{-}m}
(\rho_{RA}, 2k)$.

Now we prove the latter identity. Consider an arbitrary decoupling scheme
via removing a subsystem of size not more than $k$ qubits. Let $\sigma_{A'}$
be the catalytic state and $U:\mc{H}_{AA'}\rar \mc{H}_{A_1A_2}$ be the
unitary operation, where $A_2$ is the system to be removed and $\log |A_2|
\leq k$. We construct a decoupling scheme via random unitary operation:
the catalytic state is chosen to be $\sigma_{A'}$ and the random unitary
operation is given by
\begin{equation}
\Lambda_{AA'}(\cdot)=\frac{1}{|A_2|^2} \sum_{i=1}^{|A_2|^2}
(\1_{A_1}\ox W_{A_2}^i) U (\cdot) U^* (\1_{A_1}\ox W_{A_2}^i)^*,
\end{equation}
where $\{W_{A_2}^i\}_{i=1}^{|A_2|^2}$ are all the discrete Weyl operators
on $\mc{H}_{A_2}$. By the definition of $P_{R:A}^{\rm{dec\text{-}u}}$, for
arbitrary state $\omega_{A_1}$, we have
\begin{equation}
\label{eq:eraslardec}
\begin{split}
     &P_{R:A}^{\rm{dec\text{-}u}}(\rho_{RA}, 2k) \\
\leq &P\Big(\Lambda_{AA'}(\rho_{RA} \ox \sigma_{A'}), \rho_R
        \ox \omega_{A_1}\ox \frac{\1_{A_2}}{|A_2|}\Big) \\
 =   &P\Big(\tr_{A_2} \big[U(\rho_{RA} \ox \sigma_{A'})U^*\big] \ox
       \frac{\1_{A_2}}{|A_2|}, \rho_R \ox \omega_{A_1}
       \ox \frac{\1_{A_2}}{|A_2|}\Big) \\
 =   &P\Big(\tr_{A_2} \big[U(\rho_{RA} \ox \sigma_{A'})U^*\big],
       \rho_R \ox \omega_{A_1}\Big).
\end{split}
\end{equation}
Minimizing the last line of Eq.~(\ref{eq:eraslardec}) over $\sigma_{A'}$,
$\omega_{A_1}$ and $U:\mc{H}_{AA'}\rar \mc{H}_{A_1A_2}$, and by
Definition~\ref{def:perform-dec}, we get from Eq.~(\ref{eq:eraslardec}) that
\begin{equation}
\label{eq:eraslardec1}
P_{R:A}^{\rm{dec\text{-}u}}(\rho_{RA},2k)\leq P^{\rm{dec}}_{R:A}(\rho_{RA},k).
\end{equation}
Next, we deal with the other direction. Let $\sigma_{A'}$ and $\Lambda_{AA'}$
be the catalytic state and the random unitary channel of a decoupling scheme
via random unitary operation with cost bounded by $2k$ bits. $\Lambda_{AA'}$
can be written as
\begin{equation}
\Lambda_{AA'}(\cdot)=\frac{1}{m} \sum_{i=1}^m V_{AA'}^{i} (\cdot) (V_{AA'}^{i})^*,
\end{equation}
where $V_{AA'}^i\in\mc{U}(AA')$ and $\log m \leq 2k$. Introduce systems $C$
and $C'$ such that $|C|=|C'|=\sqrt{m}$ and denote by $\Psi_{CC'}$ the maximally
entangled state on $C$ and $C'$. We construct a decoupling scheme via removing
a subsystem: the catalytic state is given by
\[
\sigma_{A'CC'}=\frac{1}{m}\sum_{i=1}^{m}
\sigma_{A'}\ox(W_C^i\ox\1_{C'})\Psi_{CC'}(W_C^i\ox\1_{C'})^*
\]
with $\{W_C^i\}_{i=1}^{m}$ be all the discrete Weyl operators on $\mc{H}_C$,
and the unitary operation is given by
\[
U_{AA'CC'}
=\sum_{i=1}^m V_{AA'}^i \ox (W_C^i\ox\1_{C'})\Psi_{CC'}(W_C^i\ox\1_{C'})^*.
\]
Set $A_1\equiv AA'C$ be the remaining system, and $A_2\equiv C'$ be the
removed system. So, the cost is $\log|A_2|=\log|C'|\leq k$. By the
definition of $P_{R:A}^{\rm{dec}}$, for arbitrary state $\omega_{AA'}$, we
have
\begin{equation}
\label{eq:declareras}
\begin{split}
     & P_{R:A}^{\rm{dec}}(\rho_{RA}, k)\\
\leq & P\Big(\tr_{C'} \big[U_{AA'CC'}(\rho_{RA}\ox\sigma_{A'CC'})
        U_{AA'CC'}^*\big],\rho_R\ox\omega_{AA'}\ox\frac{\1_C}{|C|} \Big) \\
  =  & P\Big(\Lambda_{AA'}(\rho_{RA}\ox \sigma_{A'})\ox \frac{\1_C}{|C|},
         \rho_R \ox \omega_{AA'} \ox \frac{\1_C}{|C|} \Big)\\
  =  & P\big(\Lambda_{AA'}(\rho_{RA} \ox \sigma_{A'}),
        \rho_R\ox\omega_{AA'}\big).
\end{split}
\end{equation}
Minimizing the last line of Eq.~(\ref{eq:declareras}) over $\sigma_{A'}$,
$\Lambda_{AA'}$ and $\omega_{AA'}$, and by Definition~\ref{def:perform-dec-u},
we get from Eq.~(\ref{eq:declareras}) that
\begin{equation}
  \label{eq:declareras1}
P_{R:A}^{\rm{dec}}(\rho_{RA},k) \leq P_{R:A}^{\rm{dec\text{-}u}}(\rho_{RA},2k).
\end{equation}
The combination of Eqs.~\eqref{eq:eraslardec1} and \eqref{eq:declareras1}
completes the proof.
\end{Proof}

\subsection{Proof of the relation between decoupling and quantum state merging}
  \label{subsec:relation-dec-merg}
\begingroup
\def\thetheorem{\ref{prop:merging-dec}}
\begin{proposition}[restatement]
For a tripartite pure state $\rho_{RAB} \in \mathcal{S}({RAB})$ and
$k\geq 0$, $r \geq 0$, we have
\begin{align}
&P^{\rm{merg}}_{A \Rightarrow B}(\rho_{RAB}, k)
=P^{\rm{merg}}_{A \rightarrow B}(\rho_{RAB}, 2k)
=P^{\rm{dec}}_{R:A}(\rho_{RA}, k), \label{eqr:merg-dec-P} \\
&E^{\rm{merg}}_{A \Rightarrow B}(\rho_{RAB}, r)
=E^{\rm{merg}}_{A \rightarrow B}(\rho_{RAB}, 2r)
=E^{\rm{dec}}_{R:A}(\rho_{RA}, r). \label{eqr:merg-dec-E}
\end{align}
\end{proposition}
\addtocounter{theorem}{-1}
\endgroup

\begin{Proof}
Eq.~\eqref{eqr:merg-dec-E} follows from Eq.~\eqref{eqr:merg-dec-P}
directly. So it suffices to prove Eq.~\eqref{eqr:merg-dec-P}. Assisted
by free entanglement, a merging protocol via quantum communication of
$k$ qubits, can simulate a merging protocol via classical communication
of $2k$ bits, by using teleportation. Conversely, a merging protocol
via classical communication of $2k$ bits, can simulate a merging protocol
via quantum communication of $k$ qubits, by using dense coding. So, we
easily obtain the relation $P_{A \Rightarrow B}^{\rm{merg}}(\rho_{RAB},k)
=P_{A \rar B}^{\rm{merg}}(\rho_{RAB}, 2k)$. In the following, we prove
$P_{A \Rightarrow B}^{\text{merg}}(\rho_{RAB},k)= P_{R:A}^{\text{dec}}
(\rho_{RA}, k)$.

Consider a decoupling scheme with catalytic state $\sigma_{A'}$, in which
the system $A_2$ of size $\log|A_2|\leq k$ is removed after a unitary
operator $U_{AA'\rar A_1A_2}$ being applied. Let $\omega_{A_1}$ be an
arbitrary state on $A_1$. Let $\sigma_{A'B'}$ and $\omega_{A_1B_1}$ be
the purification of $\sigma_{A'}$ and $\omega_{A_1}$, respectively. By
Uhlmann's theorem~\cite{Uhlmann1976transition}, there exists an isometry
$V_{A_2BB' \rar ABB_1}$ such that
\begin{equation}
\label{equ:isometry}
\begin{split}
 &P\big(\tr_{A_2} \big[U_{AA'\rar A_1A_2}(\rho_{RA} \ox \sigma_{A'})
           U_{AA'\rar A_1A_2}^*\big], \rho_R \ox \omega_{A_1} \big)   \\
=&P\big(V_{A_2BB' \rightarrow ABB_1}U_{AA'\rar A_1A_2}(\rho_{RAB} \ox
   \sigma_{A'B'})U_{AA'\rar A_1A_2}^*V_{A_2BB' \rar ABB_1}^*, \rho_{RAB}
   \ox \omega_{A_1B_1}\big).
\end{split}
\end{equation}
We construct a merging protocol $\mc{M}_1$ which consists of using
shared entanglement $\sigma_{A'B'}$, Alice applying $U_{AA'\rar A_1A_2}$
and sending $A_2$ to Bob, and Bob applying isometry $V_{A_2BB' \rar ABB_1}$.
The cost of quantum communication is $\log |A_2| \leq k$. We have
\begin{equation}
  \label{eq:M1}
\mc{M}_1(\rho_{RAB})=\tr_{A_1B_1}\big[V_{A_2BB' \rar ABB_1}
    U_{AA'\rar A_1A_2}(\rho_{RAB} \ox \sigma_{A'B'})U_{AA'\rar A_1A_2}^*
    V_{A_2BB' \rar ABB_1}^*\big].
\end{equation}
By the monotonicity of purified distance under partial trace, we get
from Eqs.~\eqref{equ:isometry} and \eqref{eq:M1} that
\begin{equation}
  \label{eq:mleqdec}
     P\big(\mc{M}_1(\rho_{RAB}), \rho_{RAB}\big)
\leq P\big(\tr_{A_2}\big[U_{AA'\rar A_1A_2}(\rho_{RA} \ox \sigma_{A'})
       U_{AA'\rar A_1A_2}^*\big], \rho_R \ox \omega_{A_1} \big)
\end{equation}
Eq.~\eqref{eq:mleqdec} together with the definitions of
$P_{A \Rightarrow B}^{\text{merg}}$ and $P_{R:A}^{\text{dec}}$ implies
\begin{equation}
\label{equ:mergeleqdec}
P_{A \Rightarrow B}^{\text{merg}}(\rho_{RAB}, k)
\leq P_{R:A}^{\text{dec}}(\rho_{RA}, k).
\end{equation}

Now, we turn to the proof of the other direction. Let $\mc{M}_1$ be a
merging protocol via quantum communication of not more than $k$ qubits.
We write $\mc{M}_1$ as
\[
\mc{M}_1(\cdot)=
\tr_{A_1B_1} \big[V_{A_2BB' \rar ABB_1}U_{AA' \rar A_1A_2}\big((\cdot)
\ox \phi_{A'B'}\big)U_{AA' \rar A_1A_2}^*V_{A_2BB' \rar ABB_1}^*\big],
\]
where $\phi_{A'B'}$ is the shared entanglement, $U_{AA' \rar A_1A_2}$
is performed by Alice, $V_{A_2BB' \rar ABB_1}$ is performed by Bob, and
$A_2$ with $\log |A_2| \leq k$ is sent from Alice to Bob. Using Uhlmann's
theorem again, we know that there exists a pure state $\varphi_{A_1B_1}$
such that
\begin{equation}
\begin{split}
 &P\big(\mc{M}_1(\rho_{RAB}), \rho_{RAB}\big) \\
=&P\big(V_{A_2BB' \rar ABB_1}U_{AA' \rar A_1A_2}\big(\rho_{RAB} \ox
   \phi_{A'B'}\big)U_{AA' \rar A_1A_2}^*V_{A_2BB' \rar ABB_1}^*,
   \rho_{RAB} \ox \varphi_{A_1B_1}\big).
\end{split}
\end{equation}
Because the purified distance decreases under the action of partial trace,
we get
\begin{equation}
\label{equ:mergeuhm}
     P(\mc{M}_1(\rho_{RAB}), \rho_{RAB})
\geq P\big(\tr_{A_2}\big[U_{AA' \rar A_1A_2}\big(\rho_{RA} \ox \phi_{A'}\big)
       U_{AA' \rar A_1A_2}^*\big], \rho_R  \ox \varphi_{A_1}\big).
\end{equation}
Since $\tr_{A_2}\big[U_{AA' \rar A_1A_2}\big(\rho_{RA} \ox \phi_{A'}\big)
U_{AA' \rar A_1A_2}^*\big]$ in the right hand side of Eq.~\eqref{equ:mergeuhm}
describes a decoupling scheme via removing $\log |A_2| \leq k$ qubits,
by the definitions of $P_{A \Rightarrow B}^{\text{merg}}$ and
$P_{R:A}^{\text{dec}}$, we conclude that
\begin{equation}
\label{equ:decleqmerge}
     P_{A \Rightarrow B}^{\text{merg}}(\rho_{RAB}, k)
\geq P_{R:A}^{\text{dec}}(\rho_{RA},k).
\end{equation}
\end{Proof}

\section{Discussion}
  \label{sec:discussion}
In this work, we have characterized the reliability function of catalytic
quantum information decoupling using the sandwiched R{\'e}nyi divergence,
and have applied it to quantum state merging. Our results add new operational
meanings to the sandwiched R{\'e}nyi information quantities.

The availability of a catalytic system is crucial in our derivations.
This is not only for the characterization of the reliability function
(Theorem~\ref{thm:reliability-dec}), but also for the interplay between
different decoupling operations (Proposition~\ref{prop:dec-relation})
and the interplay between decoupling and state merging
(Proposition~\ref{prop:merging-dec}). The introduction of the catalyst
in decoupling~\cite{MBDRC2017catalytic,ADJ2017quantum} helps us solve
the problem.

Before ending the paper, we list a few open problems.
\begin{enumerate}
  \item Firstly, to characterize the reliability function of quantum information
  decoupling without involving any catalytic system is an interesting
  open problem. It is known that when only an asymptotically vanishing
  error is concerned, the best rate of cost is the same no matter a
  catalytic system is allowed or not~\cite{HOW2005partial,HOW2007quantum,
  ADHW2009mother,MBDRC2017catalytic,ADJ2017quantum}. However, for the
  reliability function we do not know whether catalysis makes a difference.

  \item Another open problem is to derive the reliability function when the rate
  of cost is above the critical value. The upper and lower bounds obtained in
  the present work do not match in the high-rate case. However, note that the
  existence of critical points in the study of reliability functions is a common
  phenomenon, where at the unsolved side the problem becomes more of a
  combinatorial feature and is hard to tackle~\cite{Gallager1968information,
  Dalai2013lower,HayashiTan2016equivocations,LYH2023tight}.

  \item Moreover, it is interesting to consider the reliability function when
  the performance is measured by the trace distance. Our method works for the
  purified distance, and equivalently, for the fidelity. But for trace distance,
  quite different techniques may be needed. It is worth mentioning that, as
  was pointed out in~\cite{Dupuis2021privacy}, for trace distance there is no
  obvious way to transfer the reliability function from quantum information
  decoupling to other quantum tasks (such as quantum state merging). This is
  in contrast to what we have done in Proposition~\ref{prop:merging-dec} and Corollary~\ref{cor:reliability-merging} which make the transference by
  employing the Uhlmann's theorem.

  \item Finally, we hope that our work will be stimulating and find more
  applications, in deriving the reliability functions for more quantum
  information tasks. Indeed, in a later work~\cite{LiYao2021reliable} we
  have applied the present results to quantum channel simulation. Can we
  extend this line of research to entanglement-assisted communication of noisy
  channels~\cite{BSST2002entanglement}, and to the treatment of the amount
  of comsumed/generated entanglement in quantum tasks? These problems are
  closely related to the size of the remaining system in quantum information
  decoupling.
\end{enumerate}

\bigskip
\appendix
\section{Boundedness of the Upper Bound $E_u(r)$}
\begin{proposition}
  \label{prop:bounded}
Let $\rho_{RA}\in\mc{S}(RA)$ be given. The function $E_u(r)=\sup\limits
_{s\geq0}\big\{s\big(r-\frac{1}{2}I_{1+s}(R:A)_\rho\big)\big\}$ is bounded
in the interval $(-\infty, \frac{1}{2}I_{\rm max}(R:A)_\rho)$.
\end{proposition}

\begin{Proof}
We show that there is a constant $C$, such that for any $\epsilon>0$,
$E_u\big(\frac{1}{2}I_{\rm max}(R:A)_\rho-\epsilon\big)\leq C$.
Proposition~\ref{prop:usym-Imax} establishes that
\beq\label{eq:bounded-1}
E_u(r)=\lim_{n\rar\infty}\frac{-1}{2n}\log
       \tr\left[\mc{E}^n\big(\rho_{AB}^{\ox n}\big)
       \Big\{\mc{E}^n\big(\rho_{AB}^{\ox n}\big)\geq 2^{n\cdot2r}
       \rho_A^{\ox n}\ox\omega^{(n)}_{B^n}\Big\} \right],
\eeq
where $\omega^{(n)}_{B^n}$ is the universal symmetric state of
Lemma~\ref{lemma:unisym} and $\mc{E}^n$ is the pinching map associated
with $\rho_A^{\ox n}\ox\omega^{(n)}_{B^n}$. Setting $M_A=\rho_A$ and
letting $s\rar\infty$ in Lemma~\ref{lem:additive-Is}, we get
\[
\frac{1}{n}D_{\rm max}\Big(\mc{E}^n(\rho_{AB}^{\ox n})\big\|\rho_A^{\ox n}
\ox\omega^{(n)}_{B^n}\Big) = I_{\rm max}(A:B)_\rho+O\Big(\frac{\log n}{n}\Big).
\]
This implies that, for arbitrary $\epsilon>0$, there exists a common
eigenvector $\ket{\varphi_n}$ of $\mc{E}^n(\rho_{AB}^{\ox n})$ and
$\rho_A^{\ox n}\ox\omega^{(n)}_{B^n}$, such that for $n$ big enough,
\beq\label{eq:bounded-2}
\Big\{\mc{E}^n\big(\rho_{AB}^{\ox n}\big)\geq 2^{n(I_{\rm max}(A:B)_\rho
-2\epsilon)}\rho_A^{\ox n}\ox\omega^{(n)}_{B^n}\Big\}
\geq \proj{\varphi_n}
\eeq
and
\beq\begin{split}\label{eq:bounded-3}
\bra{\varphi_n}\mc{E}^n\big(\rho_{AB}^{\ox n}\big)\ket{\varphi_n}
&\geq 2^{n(I_{\rm max}(A:B)_\rho-2\epsilon)}
 \bra{\varphi_n}\rho_A^{\ox n}\ox\omega^{(n)}_{B^n}\ket{\varphi_n} \\
&\geq 2^{n(I_{\rm max}(A:B)_\rho-2\epsilon)}
      \lambda_{\rm min}\big(\rho_A^{\ox n}\ox\omega^{(n)}_{B^n}\big)\\
&\geq 2^{n(I_{\rm max}(A:B)_\rho-2\epsilon)}
      \big(\lambda_{\rm min}(\rho_A)\big)^n\frac{1}{g_{n,|B|}\,|B|^n}.
\end{split}\eeq
In Eq.~\eqref{eq:bounded-3}, we have used the fact
$g_{n,|B|}\omega^{(n)}_{B^n}\geq(\frac{\1_B}{|B|})^{\ox n}$ and
$\lambda_{\rm min}(X)$ denotes the minimal eigenvalue of $X$. Combining
Eq.~\eqref{eq:bounded-1}, Eq.~\eqref{eq:bounded-2} and
Eq.~\eqref{eq:bounded-3}, we obtain
\beq\begin{split}
E_u\big(\frac{1}{2}I_{\rm max}(A:B)_\rho-\epsilon\big)
&\leq \lim_{n\rar\infty}\frac{-1}{2n}\log
     \bra{\varphi_n}\mc{E}^n\big(\rho_{AB}^{\ox n}\big)\ket{\varphi_n} \\
&\leq \frac{1}{2}\log\frac{|B|}{\lambda_{\rm min}(\rho_A)}
      -\frac{1}{2}I_{\rm max}(A:B)_\rho+\epsilon.
\end{split}\eeq
At last, since $E_u(r)$ is monotonically increasing, we can choose
$C=\frac{1}{2}\log\frac{|B|}{\lambda_{\rm min}(\rho_A)}
-\frac{1}{2}I_{\rm max}(A:B)_\rho$.
\end{Proof}

\acknowledgments
The authors are grateful to Mark Wilde for bringing to their attention References~\cite{WWW2019quantifying,KhatriWilde2020principles,Sharma2015random},
as well as to Masahito Hayashi for pointing out \cite{Hayashi2016security}
and to an anonymous referee for pointing out \cite{MBGYA2019a}.
They further thank Nilanjana Datta, Masahito Hayashi, Marco Tomamichel, Mark
Wilde, Xiao Xiong and Dong Yang for comments or discussions on related topics.
The research of KL was supported by the National Natural Science Foundation of
China (Nos. 61871156, 12031004), and the research of YY was supported by the
National Natural Science Foundation of China (Nos. 61871156, 12071099).


\begin{thebibliography}{99}

\bibitem{HOW2005partial}
Horodecki, M., Oppenheim, J., Winter, A.:
Partial quantum information.
Nature {\bf 436}, 673 (2005)

\bibitem{HOW2007quantum}
Horodecki, M., Oppenheim, J., Winter, A.:
Quantum state merging and negative information.
Commun. Math. Phys. {\bf 269}(1), 107--136 (2007)

\bibitem{ADHW2009mother}
Abeyesinghe, A., Devetak, I., Hayden, P., Winter, A.:
The mother of all protocols: restructuring quantum information's family tree.
Proc. R. Soc. A {\bf 465}, 2537--2563 (2009)

\bibitem{DevetakYard2008exact}
Devetak, I., Yard, J.:
Exact cost of redistributing multipartite quantum states.
Phys. Rev. Lett. {\bf 100}, 230501 (2008)

\bibitem{YardDevetak2009optimal}
Yard, J., Devetak, I.:
Optimal quantum source coding with quantum side information at the encoder and decoder.
IEEE Trans. Inf. Theory {\bf 55}(11), 5339--5351 (2009)

\bibitem{BDHSW2014quantum}
Bennett, C.H., Devetak, I., Harrow, A.W., Shor, P.W., Winter, A.:
The quantum reverse {Shannon} theorem and resource tradeoffs for simulating quantum channels.
IEEE Trans. Inf. Theory {\bf 60}(5), 2926--2959 (2014)

\bibitem{BCR2011the}
Berta, M., Christandl, M., Renner, R.:
The quantum reverse {Shannon} theorem based on one-shot information theory.
Commun. Math. Phys. {\bf 306}(3), 579--615 (2011)

\bibitem{BBMW2018conditional}
Berta, M., Brand{\~a}o, F.G.S.L., Majenz, C., Wilde, M.M.:
Conditional decoupling of quantum information.
Phys. Rev. Lett. {\bf 121}, 040504 (2018)

\bibitem{BFW2013quantum}
Berta, M., Fawzi, O., Wehner, S.:
Quantum to classical randomness extractors.
IEEE Trans. Inf. Theory {\bf 60}(2), 1168--1192 (2013)

\bibitem{DARDV2011thermodynamic}
del Rio, L., {\AA}berg, J., Renner, R., Dahlsten, O., Vedral, V.:
The thermodynamic meaning of negative entropy.
Nature {\bf 474} 61 (2011)

\bibitem{BrandaoHorodecki2013area}
Brand{\~a}o, F.G.S.L., Horodecki, M.:
An area law for entanglement from exponential decay of correlations.
Nat. Phys. {\bf 9}, 721 (2013)

\bibitem{BrandaoHorodecki2015exponential}
Brand{\~a}o, F.G.S.L., Horodecki, M.:
Exponential decay of correlations implies area law.
Commun. Math. Phys. {\bf 333}(2), 761--798 (2015)

\bibitem{Aberg2013truly}
{\AA}berg, J.:
Truly work-like work extraction via a single-shot analysis.
Nat. Commun. {\bf 4}, 1925 (2013)

\bibitem{DHRW2016relative}
del Rio, L., Hutter, A., Renner, R., Wehner, S.:
Relative thermalization.
Phys. Rev. E {\bf 94}, 022104 (2016)

\bibitem{HP2007black}
Hayden, P., Preskill, J.:
Black holes as mirrors: quantum information in random subsystems.
JHEP {\bf 09}, 120 (2007)

\bibitem{BraunsteinPati2007quantum}
Braunstein, S.L., Pati, A.K.:
Quantum information cannot be completely hidden in correlations: implications for the black-hole information paradox.
Phys. Rev. Lett. {\bf 98}, 080502 (2007)

\bibitem{BSZ2013better}
Braunstein, S.L., Pirandola, S., {\.Z}yczkowski, K.:
Better late than never: information retrieval from black holes.
Phys. Rev. Lett. {\bf 110}, 101301 (2013)

\bibitem{KamilChristoph2016one}
Br\'adler, K., Adami, C.:
One-shot decoupling and page curves from a dynamical model for black hole evaporation.
Phys. Rev. Lett. {\bf 116}, 101301 (2016)

\bibitem{DBWR2014one}
Dupuis, F., Berta, M., Wullschleger, J., Renner, R.:
One-shot decoupling.
Commun. Math. Phys. {\bf 328}(1), 251--284 (2014)

\bibitem{ADJ2017quantum}
Anshu, A., Devabathini, V.K., Jain, R.:
Quantum communication using coherent rejection sampling.
Phys. Rev. Lett. {\bf 119}, 120506 (2017)

\bibitem{MBDRC2017catalytic}
Majenz, C., Berta, M., Dupuis, F., Renner, R., Christandl, M.:
Catalytic decoupling of quantum information.
Phys. Rev. Lett. {\bf 118}, 080503 (2017)

\bibitem{MBGYA2019a}
Mojahedian, M.M., Beigi, S., Gohari, A., Yassaee, M.H., Aref, M.R.:
A correlation measure based on vector-valued $L_p$-norms.
IEEE Trans. Inf. Theory {\bf 65}(12), 7985--8004 (2019)

\bibitem{WakakuwaNakata2021one}
Wakakuwa, E., Nakata, Y.:
One-shot randomized and nonrandomized partial decoupling.
Commun. Math. Phys. {\bf 386}(2), 589--649 (2021)

\bibitem{Dupuis2021privacy}
Dupuis, F.:
Privacy amplification and decoupling without smoothing.
IEEE Trans. Inf. Theory {\bf 69}(12), 7784--7792 (2023)

\bibitem{SDTR2013decoupling}
Szehr, O., Dupuis, F., Tomamichel, M., Renner, R.:
Decoupling with unitary approximate two-designs.
New J. Phys. {\bf 15}, 053022 (2013)

\bibitem{BrownFawzi2015decoupling}
Brown, W., Fawzi, O.:
Decoupling with random quantum circuits.
Commun. Math. Phys. {\bf 340}(3), 867--900 (2015)

\bibitem{NHMW2017decoupling}
Nakata, Y., Hirche, C., Morgan, C., Winter, A.:
Decoupling with random diagonal unitaries.
Quantum {\bf 1}, 18 (2017)

\bibitem{Sharma2015random}
Sharma, N.:
Random coding exponents galore via decoupling (2015).
arXiv:1504.07075

\bibitem{ABJT2020partially}
Anshu, A., Berta, M., Jain, R., Tomamichel, M.:
Partially smoothed information measures.
IEEE Trans. Inf. Theory {\bf 66}(8), 5022--5036 (2020)

\bibitem{Shannon1959probability}
Shannon, C.E.:
Probability of error for optimal codes in a {Gaussian} channel.
Bell Syst. Tech. J. {\bf 38}(3), 611--656 (1959)

\bibitem{Gallager1968information}
Gallager, R.:
Information Theory and Reliable Communication.
John Wiley \& Sons, New York (1968)

\bibitem{BurnashevHolevo1998on}
Burnashev, M., Holevo, A.S.:
On the reliability function for a quantum communication channel.
Probl. Inf. Transm. {\bf 34}(2), 97--107 (1998)

\bibitem{Holevo2000reliability}
Holevo, A.S.:
Reliability function of general classical-quantum channel.
IEEE Trans. Inf. Theory {\bf 46}(6), 2256--2261 (2000)

\bibitem{Winter1999coding}
Winter, A.:
Coding theorems of quantum information theory.
PhD Thesis, Universit\"{a}t Bielefeld (1999)

\bibitem{Dalai2013lower}
Dalai, M.:
Lower bounds on the probability of error for classical and classical-quantum channels.
IEEE Trans. Inf. Theory {\bf 59}(12), 8027--8056 (2013)

\bibitem{Hayashi2015precise}
Hayashi, M.:
Precise evaluation of leaked information with secure randomness extraction in the presence of quantum attacker.
Commun. Math. Phys. {\bf 333}(1), 335--350 (2015)

\bibitem{DalaiWinter2017constant}
Dalai, M., Winter, A.:
Constant compositions in the sphere packing bound for classical-quantum channels.
IEEE Trans. Inf. Theory {\bf 63}(9), 5603--5617 (2017)

\bibitem{CHT2019quantum}
Cheng, H.-C., Hsieh, M.-H., Tomamichel, M.:
Quantum sphere-packing bounds with polynomial prefactors.
IEEE Trans. Inf. Theory {\bf 65}(5), 2872--2898 (2019)

\bibitem{CHDH2020non}
Cheng, H.-C., Hanson, E.P., Datta, N., Hsieh, M.-H.:
Non-asymptotic classical data compression with quantum side information.
IEEE Trans. Inf. Theory {\bf 67}(2), 902--930 (2020)

\bibitem{KoenigWehner2009strong}
Koenig, R., Wehner, S.:
A strong converse for classical channel coding using entangled inputs.
Phys. Rev. Lett. {\bf 103}, 070504 (2009)

\bibitem{SharmaWarsi2013fundamental}
Sharma, N., Warsi, N.A.:
Fundamental bound on the reliability of quantum information transmission.
Phys. Rev. Lett. {\bf 110}, 080501 (2013)

\bibitem{WWY2014strong}
Wilde, M.K., Winter, A., Yang, D.:
Strong converse for the classical capacity of entanglement-breaking
and {Hadamard} channels via a sandwiched {R{\'e}nyi} relative entropy.
Commun. Math. Phys. {\bf 331}(2), 593--622 (2014)

\bibitem{GuptaWilde2015multiplicativity}
Gupta, M.K., Wilde, M.M.:
Multiplicativity of completely bounded $p$-norms implies a strong
converse for entanglement-assisted capacity.
Commun. Math. Phys. {\bf 334}(2), 867--887 (2015)

\bibitem{CMW2016strong}
Cooney, T., Mosonyi, M., Wilde, M.M.:
Strong converse exponents for a quantum channel discrimination
problem and quantum-feedback-assisted communication.
Commun. Math. Phys. {\bf 344}(3), 797--829 (2016)

\bibitem{MosonyiOgawa2017strong}
Mosonyi, M., Ogawa, T.:
Strong converse exponent for classical--quantum channel coding.
Commun. Math. Phys. {\bf 355}(1), 373--426 (2017)

\bibitem{MDSFT2013on}
M{\"u}ller-Lennert, M., Dupuis, F., Szehr, O., Fehr, S., Tomamichel, M.:
On quantum {R{\'e}nyi} entropies: a new generalization and some properties.
J. Math. Phys. {\bf 54}, 122203 (2013)

\bibitem{LYH2023tight}
Li, K., Yao, Y., Hayashi, M.:
Tight exponential analysis for smoothing the max-relative entropy and
for quantum privacy amplification.
IEEE Trans. Inf. Theory {\bf 69}(3), 1680--1694 (2023)

\bibitem{MosonyiOgawa2015quantum}
Mosonyi, M., Ogawa, T.:
Quantum hypothesis testing and the operational interpretation of the
quantum {R{\'e}nyi} relative entropies.
Commun. Math. Phys. {\bf 334}(3), 1617--1648 (2015)

\bibitem{MosonyiOgawa2015two}
Mosonyi, M., Ogawa, T.:
Two approaches to obtain the strong converse exponent of quantum
hypothesis testing for general sequences of quantum states.
IEEE Trans. Inf. Theory {\bf 61}(12), 6975--6994 (2015)

\bibitem{HayashiTomamichel2016correlation}
Hayashi, M., Tomamichel, M.:
Correlation detection and an operational interpretation of the
{R{\'e}nyi} mutual information.
J. Math. Phys. {\bf 57}, 102201 (2016)

\bibitem{GLN2005distance}
Gilchrist, A., Langford, N.K., Nielsen, M.A.:
Distance measures to compare real and ideal quantum processes.
Phys. Rev. A {\bf 71}, 062310 (2005)

\bibitem{TCR2009fully}
Tomamichel, M., Colbeck, R., Renner, R.:
A fully quantum asymptotic equipartition property.
IEEE Trans. Inf. Theory {\bf 55}(12), 5840--5847 (2009)

\bibitem{Uhlmann1976transition}
Uhlmann, A.:
The ``transition probability'' in the state space of a $^\ast$-algebra.
Rep. Math. Phys. {\bf 9}(2), 273--279 (1976)

\bibitem{Stinespring1955positive}
Stinespring, F.:
Positive functions on C$^*$-algebras.
Proc. Amer. Math. Soc. {\bf 6}(2), 211--216 (1955)

\bibitem{Hayashi2002optimal}
Hayashi, M.:
Optimal sequence of quantum measurements in the sense of {Stein's}
lemma in quantum hypothesis testing.
J. Phys. A: Math. Gen. {\bf 35}, 10759 (2002)

\bibitem{CKR2009postselection}
Christandl, M., K{\"o}nig, R., Renner, R.:
Postselection technique for quantum channels with applications to quantum cryptography.
Phys. Rev. Lett. {\bf 102}, 020504 (2009)

\bibitem{Hayashi2009universal}
Hayashi, M.:
Universal coding for classical--quantum channel.
Commun. Math. Phys. {\bf 289}(3), 1087--1098 (2009)

\bibitem{Beigi2013sandwiched}
Beigi, S.:
Sandwiched {R{\'e}nyi} divergence satisfies data processing inequality.
J. Math. Phys. {\bf 54}, 122202 (2013)

\bibitem{Umegaki1954conditional}
Umegaki, H.:
Conditional expectation in an operator algebra.
Tohoku Math. J. {\bf 6}(2), 177--181 (1954)

\bibitem{Datta2009min}
Datta, N.:
Min- and max-relative entropies and a new entanglement monotone.
IEEE Trans. Inf. Theory {\bf 55}(6), 2816--2826 (2009)

\bibitem{CBR2013smooth}
Ciganovi{\'c}, N., Beaudry, N.J., Renner, R.:
Smooth max-information as one-shot generalization for mutual information.
IEEE Trans. Inf. Theory {\bf 60}(3), 1573--1581 (2013)

\bibitem{GPW2005quantum}
Groisman, B., Popescu, S., Winter, A.:
Quantum, classical, and total amount of correlations in a quantum state.
Phys. Rev. A {\bf 72}, 032317 (2005)

\bibitem{HayashiTan2016equivocations}
Hayashi, M., Tan, V.Y.F.:
Equivocations, exponents, and second-order coding rates under various
{R{\'e}nyi} information measures.
IEEE Trans. Inf. Theory {\bf 63}(2), 975--1005 (2016)

\bibitem{WWW2019quantifying}
Wang, K., Wang, X., Wilde, M.M.:
Quantifying the unextendibility of entanglement.
arXiv:1911.07433, 2019.

\bibitem{KhatriWilde2020principles}
Khatri, S., Wilde, M.M.:
Principles of quantum communication theory: a modern approach (2020).
arXiv:2011.04672

\bibitem{AJW2017generalized}
Anshu, A., Jain, R., Warsi, N.A.:
A generalized quantum {Slepian}--{Wolf}.
IEEE Trans. Inf. Theory {\bf 64}(3), 1436--1453 (2017)

\bibitem{AJW2018building}
Anshu, A., Jain, R., Warsi, N.A.:
Building blocks for communication over noisy quantum networks.
IEEE Trans. Inf. Theory {\bf 65}(2), 1287--1306 (2018)

\bibitem{AJW2019convex}
Anshu, A., Jain, R., Warsi, N.A.:
Convex-split and hypothesis testing approach to one-shot quantum
measurement compression and randomness extraction.
IEEE Trans. Inf. Theory {\bf 65}(9), 5905--5924 (2019)

\bibitem{AHJ2018quantifying}
Anshu, A., Hsieh, M.-H., Jain, R.:
Quantifying resources in general resource theory with catalysts.
Phys. Rev. Lett. {\bf 121}, 190504 (2018)

\bibitem{BertaMajenz2018disentanglement}
Berta, M., Majenz, C.:
Disentanglement cost of quantum states.
Phys. Rev. Lett. {\bf 121}, 190503 (2018)

\bibitem{FBB2021thermodynamic}
Faist, P., Berta, M., Brandao, F.G.S.L.:
Thermodynamic implementations of quantum processes.
Commun. Math. Phys. {\bf 384}(3), 1709--1750 (2021)

\bibitem{LipkaSkrzypczyk2021all}
Lipka-Bartosik, P., Skrzypczyk, P.:
All states are universal catalysts in quantum thermodynamics.
Phys. Rev. X {\bf 11}, 011061 (2021)

\bibitem{FrankLieb2013monotonicity}
Frank, R.L., Lieb, E.H.:
Monotonicity of a relative {R{\'e}nyi} entropy.
J. Math. Phys. {\bf 54}, 122201 (2013)

\bibitem{Renner2005security}
Renner, R.:
Security of quantum key distribution.
PhD Thesis, ETH Zurich (2005)

\bibitem{Hayashi2016security}
Hayashi, M.:
Security analysis of $\varepsilon$-almost dual $\text{universal}_2$ hash
functions: smoothing of min entropy versus smoothing of {R\'enyi} entropy of order 2.
IEEE Trans. Inf. Theory {\bf 62}(6), 3451--3476 (2016)

\bibitem{Chen2000generalization}
Chen, P.-N.:
Generalization of {G\"artner--Ellis} theorem.
IEEE Trans. Inf. Theory {\bf 46}(7), 2752--2760 (2000)

\bibitem{LiYao2021reliable}
Li, K., Yao, Y.:
Reliable simulation of quantum channels: the error exponent (2021).
arXiv:2112.04475

\bibitem{BSST2002entanglement}
Bennett, C.H., Shor, P.W., Smolin, J.A., Thapliyal, A.V.:
Entanglement-assisted capacity of a quantum channel and the reverse {Shannon} theorem.
IEEE Trans. Inf. Theory {\bf 48}(10), 2637--2655 (2002)

\end{thebibliography}
\end{document}